\documentclass[11pt]{article}
\usepackage{amsmath,amssymb,amsfonts,latexsym,graphicx,amsthm}
\usepackage{fullpage,color}
\usepackage{url}
\usepackage{hyperref}
\usepackage{comment}
\usepackage[linesnumbered,boxed,ruled,vlined]{algorithm2e}
\usepackage{framed}
\usepackage[shortlabels]{enumitem}
\usepackage{cleveref}
\usepackage{todonotes}
\usepackage[normalem]{ulem}
\usepackage{mathabx}
\usepackage{tikz}
\usetikzlibrary{calc}
\usetikzlibrary{arrows.meta}
\usetikzlibrary{backgrounds}
\usetikzlibrary{arrows.meta}
\usetikzlibrary{positioning}
\usetikzlibrary{decorations.pathreplacing}

\pgfdeclarelayer{foreground}
\pgfsetlayers{main,foreground}

\usepackage{subcaption}
\usepackage[margin=1in]{geometry}
\usepackage{ifthen}
\usepackage{thm-restate}
\usepackage{xcolor}

\usepackage{stmaryrd}

\newtheorem{theorem}{Theorem}[section]

\newtheorem{lemma}{Lemma}[section]

\newtheorem{corollary}{Corollary}[section]

\newtheorem{definition}{Definition}[section]

\newtheorem{question}{Question}[section]

\newtheorem{observation}{Observation}[section]

        {\medskip}

\newcommand{\eps}{\epsilon}
\newcommand{\hide}[1]{}

\newcommand{\ceil}[1]{\left\lceil #1 \right\rceil}

\definecolor{BrickRed}{rgb}{0.8, 0.25, 0.33}
\def\EMPH#1{\emph{\textcolor{BrickRed}{#1}}}

\newcommand{\E}{\mathbb{E}}
\newcommand{\Prob}{\mathbb{P}}
\newcommand{\outdeg}{d^{+}}
\newcommand{\wt}{\operatorname{wt}}

\title{Dynamic Edge Orientation via Random Walks: \\From Trees to Outerplanar Graphs and Beyond}
\author{
Gabriel Marques Domingues
\thanks{Tel Aviv University,
\href{mailto:gm@mail.tau.ac.il}{gm@mail.tau.ac.il}.
Gabriel Marques Domingues is supported by the Israel Science Foundation,
grant No.~1948/21.}
\and
Minh Hang Nguyen
\thanks{Tel Aviv University,
\href{mailto:nguyenminhhang31198@gmail.com}{nguyenminhhang31198@gmail.com}.
Minh Hang Nguyen is funded by the European Union
(ERC, DynOpt, 101043159).}
\and
Shay Solomon
\thanks{Tel Aviv University,
\href{mailto:shayso@tauex.tau.ac.il}{shayso@tauex.tau.ac.il}.
Shay Solomon is funded by the European Union
(ERC, DynOpt, 101043159). Views and opinions expressed are however
those of the author(s) only and do not necessarily reflect those of
the European Union or the European Research Council. Neither the
European Union nor the granting authority can be held responsible
for them. Shay Solomon is also funded by a grant from the
United States-Israel Binational Science Foundation (BSF), Jerusalem,
Israel, and the United States National Science Foundation (NSF).}
}
\date{}

\begin{document}
\begin{titlepage}
\maketitle
\vspace{-1.5em}

\begin{abstract}
We study the \emph{fully dynamic edge orientation problem}, focusing on \EMPH{worst-case} time bounds. 
An undirected graph undergoes edge insertions and deletions, and the goal is to maintain an orientation with small {\em maximum outdegree} (hereafter, outdegree) and small worst-case update time. The  outdegree of any orientation is at least $\alpha-1$, where $\alpha$ is the graph's \EMPH{arboricity}, i.e., the minimum number of forests into which its edge set can be partitioned. 
When $\alpha = O(1)$, it is long known that both the outdegree and the worst-case update time can be bounded by $O(\log n)$. Despite numerous follow-ups, no $o(\log^3 n)$ worst-case update time is known for maintaining constant outdegree, even for very basic graph families---with a notable exception, \EMPH{forests}.

For forests, a \EMPH{simple folklore} algorithm maintains outdegree 2 via \EMPH{random walks}: When an insertion creates a vertex of outdegree 3, the algorithm repeatedly chooses a uniformly random outgoing edge until reaching a vertex of outdegree at most 1, and then flips the resulting directed path. 
As the underlying graph is cycle-free, the path length is easily shown to be $O(\log n)$ in expectation, and also with high probability for polynomially long update sequences.

We prove that this simple random walk paradigm extends to \EMPH{outerplanar graphs}. Our algorithm maintains constant outdegree with $O(\log n)$ worst-case update time, where the time bound holds in expectation, and also with high probability for polynomially long update sequences.
We give a \EMPH{tight analysis}: outdegree 4 is achievable with $O(\log n)$-length paths, while outdegree 3 incurs $\mathtt{poly}(n)$-length paths.  
We also extend the argument to $K_{2,t}$-minor-free graphs, for any $t \ge 2$, with the outdegree bound depending only on $t$ and with the same update time guarantees. The locality of the random walk rule also yields an essentially lossless implementation in \EMPH{dynamic distributed networks}, with the same asymptotic bounds on rounds and messages.

Finally, we establish a limitation of this paradigm. For every fixed $\Delta \geq3$, there are $n$-vertex oriented 
\EMPH{series-parallel} graphs (which have treewidth at most two and are in particular planar), in which a uniform directed random walk from a designated vertex requires 
$n^{\Omega_\Delta(1)}$ steps, both in expectation and with high probability, to reach a vertex of outdegree below $\Delta$.
Thus planarity and even treewidth 2 alone do not guarantee a short hitting time for the random walk paradigm.
\end{abstract}

\thispagestyle{empty}
\end{titlepage}

\section{Introduction}
\label{sec:introduction}

Consider an \EMPH{orientation} of an undirected graph $G=(V,E)$ whose \emph{maximum outdegree}, hereafter simply \emph{outdegree}, is $\Delta$.  If $\Delta$ is small, this gives a simple and
memory-efficient graph representation: an adjacency query for two
vertices $u$ and $v$ can be answered in $O(\Delta)$ worst-case time by
scanning their two out-neighbor lists.  This motivated the pioneering
work of Brodal and Fagerberg~\cite{BF99}: can such a basic representation
be maintained efficiently as edges are inserted and deleted?

In the resulting \emph{fully dynamic edge orientation problem}, we
maintain an orientation under edge updates while keeping both its
outdegree and update time small.  Update time may be measured
\EMPH{amortized} over the entire sequence or in the \EMPH{worst case}
per update.  We also measure the \emph{recourse}, the number of edge
flips per update, again either amortized or in the worst case. The natural benchmark for the outdegree is the graph's \emph{arboricity} $\alpha(G)$, defined as the minimum number of forests into which its edges can be partitioned.  By the Nash--Williams theorem~\cite{NashWilliams64},
$\alpha(G)=\max_{U\subseteq V,\,|U|\ge2}
  \left\lceil \frac{|E(G[U])|}{|U|-1}\right\rceil$.
Thus arboricity measures the graph's {\em uniform sparsity}.
Forests, outerplanar graphs, and planar graphs have arboricity at most $1$, $2$, and $3$, respectively, while every family excluding a fixed minor has bounded arboricity.  Every graph admits an orientation of outdegree at most
$\alpha(G)$, whereas no orientation can have outdegree below
$\alpha(G)-1$.  Thus, up to one unit, arboricity determines the smallest outdegree one can hope to maintain.  Throughout, $\alpha$ denotes a fixed upper bound on the arboricity over the entire update sequence.

The algorithm of Brodal and Fagerberg~\cite{BF99}, hereafter the
\emph{BF algorithm}, maintains  outdegree $O(\alpha)$ with
$O(\log n)$ amortized update time.  Its update time is linear in the
number of flips, so it
has $O(\log n)$ amortized recourse.  The algorithm
itself is very simple.  Whenever a vertex exceeds an outdegree threshold
$\Delta=O(\alpha)$, the algorithm \emph{resets} it by flipping all its
outgoing edges.  Any vertex whose outdegree now exceeds $\Delta$ is
placed in a global queue and reset in turn.

It was shown in~\cite{Kowalik07}, and later in~\cite{HTZ14}, that the BF algorithm yields a general tradeoff between
outdegree and amortized update time.  For constant arboricity, it
provides outdegree $O(\Delta)$ with amortized update time
$O(\log n/\Delta)$.
Whether this tradeoff can be improved at any point of the tradeoff curve has remained an
outstanding open question since~\cite{BF99}, \EMPH{even for forests}.

\vspace{-5pt}
\refstepcounter{subsection}\label{sec:intro-worst-case}
\paragraph{1.1~ Worst-case update time.~}
The BF tradeoff gives worst-case adjacency-query time, but only
amortized update time: a single update may trigger a long cascade of
resets.  BF explicitly asked whether comparable bounds
can be obtained with worst-case update time~\cite{BF99}.

A positive resolution would have implications well beyond adjacency
queries: dynamic low-outdegree orientations are a basic primitive for
dynamic distance oracles~\cite{KK03},
$k$-clique counting~\cite{chiba1985arboricity,Epp94,DLSY20},
maximal matching~\cite{NS13,HTZ14,KKPS14,CCHHQRS24},
maximal independent set~\cite{OSSW18},
and graph coloring~\cite{SW18,henzinger2020explicit,CR22,christiansen2023improved},
to name a few notable examples;
see \Cref{sec:intro-related} for further related work.
Since constant-arboricity graph families are the focus of this paper, we next state the central question in that regime.
\vspace{-3pt}
\begin{question}
\label{q:intro-worst-case}
Can one maintain constant outdegree in fully dynamic graphs of
constant arboricity using $O(\log n)$ \EMPH{worst-case} update time?
\end{question}
\vspace{-5pt}

If the entire update sequence is known in advance, often referred to as the
\emph{offline setting}, then the analogous worst-case recourse question
is already settled.  Brodal and Fagerberg~\cite{BF99} showed that, for
every update sequence of arboricity at most $\alpha$ and every
$\Delta>\alpha$, there exists a sequence of $\Delta$-orientations in
which each update flips at most
$\lceil\log_{\Delta/\alpha}n\rceil$ edges.  Thus, for constant $\alpha$
and $\Delta=\alpha+1$, $O(\log n)$ worst-case recourse is always
achievable offline.  This existential guarantee, however, does not yield
a fast update procedure without knowledge of future updates.

The first nontrivial worst-case update time bounds were obtained in~\cite{KKPS14}.  For constant
arboricity, their algorithm uses a simple local invariant to maintain
 outdegree $O(\log n)$ with $O(\log n)$ worst-case update time.
Thus it matches the BF algorithm's logarithmic time bound, but loses a
logarithmic factor in the outdegree.

A simple greedy algorithm~\cite{BerglinBrodal20} refined the
tradeoff: for constant arboricity, it achieves outdegree $O(\log n)$
with $O(\sqrt{\log n})$ worst-case update time.
Subsequent work pushed the outdegree much closer to optimal.
Outdegree $(1+\eps)\alpha+2$ with
$O(\eps^{-6}\alpha^2\log^3 n)$ worst-case update time was achieved in~\cite{CR22}, and the dependence on $\alpha$ in the update time was later improved to $O(\eps^{-6}\log^3 n\log\alpha)$ in~\cite{CCHHQRS24}. 
For constant $\alpha$ and fixed $\eps>0$, both bounds are $O(\log^3 n)$.
Thus, for constant outdegree, the best known worst-case update time is
$O(\log^3 n)$, and \Cref{q:intro-worst-case} remains open.

There is, however, one clean exception beyond the trivial case of
bounded-degree graphs: \EMPH{forests}.  A \EMPH{simple folklore} random walk
algorithm maintains outdegree $2$ with $O(\log n)$ expected update time,
and with $O(\log n)$ update time throughout every polynomial-length
update sequence w.h.p.  This suggests a particularly simple approach to
\Cref{q:intro-worst-case}, and is the starting point of our work.

\vspace{-5pt}
\refstepcounter{subsection}\label{sec:intro-forests}
\paragraph{1.2~ The random walk paradigm in forests.~}

The simple folklore algorithm maintains outdegree $2$ as follows.  Orient
each inserted edge arbitrarily.  If the insertion raises its tail's
outdegree to $3$, start a uniform directed random walk there, stop at the first vertex of
outdegree $\le 1$, and flip the resulting directed path. The starting vertex loses one outgoing edge, the terminal vertex gains
one, and every internal outdegree is unchanged, restoring outdegree $2$.  Deletions require no
repair.

The analysis is immediate since a forest contains no cycle, so the
walk is a simple path.  Until it stops, every visited vertex has at least
two outgoing edges; hence any fixed $k$-edge path is followed with
probability at most $2^{-k}$.  Since there is at most one simple path from the starting vertex to each
possible endpoint, if $\tau$ denotes the walk length then
$\Pr(\tau>k)\le\min\{1,n2^{-k}\}$, and hence
$\E[\tau]=O(\log n)$.  Traversing and flipping the path take
$O(\tau+1)$ time and exactly $\tau$ flips, so every update has
$O(\log n)$ expected update time and recourse.  Importantly, this bound
holds even after conditioning on the complete preceding history, and
hence against an adaptive adversary.  A union bound then gives
$O(\log n)$ update time and recourse for every update simultaneously
throughout any polynomial-length execution w.h.p.
We stress that there always (with probability 1) \EMPH{exists} an $O(\log n)$-length path on a forest that reaches a vertex of outdegree less than $2$. The randomization is only needed to efficiently find such a path, which is done via a random walk.

Several known algorithms are themselves notably
simple~\cite{BF99,KKPS14,BerglinBrodal20}. The main advantage of the random walk paradigm over previous algorithms is that it is essentially \EMPH{structure-free}: it only samples outgoing
edges, records the resulting path, and flips it, using no global data structures or auxiliary invariants. This feature of the algorithm yields an essentially lossless implementation in the
\EMPH{dynamic distributed setting}; see \Cref{cor:intro-distributed}.

Random walks were used in {\em incremental forests} to obtain outdegree $3$ with $O(\log\log n)$ flips per insertion and $O(\log n\log\log n)$ worst-case insertion time w.h.p.~\cite{BKKPS21}.  Their \emph{Dancing-Walk} algorithm is more involved and exploits the absence of deletions.  Although their paper does not state the adversary model explicitly, their analysis fixes the insertion sequence and hence establishes the
guarantee against an \emph{oblivious adversary}. By contrast, the conditional tail bound above applies against an
\emph{adaptive adversary} that observes the current orientation and all
past random choices; our $O(\log n)$ bounds hold simultaneously over any
polynomial-length execution w.h.p.

We ask if this simple approach already resolves
\Cref{q:intro-worst-case} for a basic graph family beyond forests.
Outerplanar graphs are a natural first test: they are the graphs with no $K_4$ or $K_{2,3}$ minor, equivalently, the graphs admitting a planar drawing in which every vertex lies on the outer face.  They have arboricity at most 2, but unlike forests they may contain many cycles.

\vspace{-5pt}
\begin{question}
\label{q:intro-main}
Can the uniform random walk paradigm maintain constant outdegree in
fully dynamic outerplanar graphs, against an adaptive adversary, with
$O(\log n)$ expected update time per update and $O(\log n)$ update time
throughout every polynomial-length update sequence w.h.p.?
\end{question}
\vspace{-14pt}
\refstepcounter{subsection}\label{sec:intro-contributions}
\paragraph{1.3~ Our contributions.~}

We answer \Cref{q:intro-main} affirmatively.
Fix a threshold $\Delta$.  Orient each inserted edge arbitrarily; if its tail now has outdegree $\Delta+1$, start a uniform directed random walk there and stop at the first vertex of outdegree below $\Delta$.  In contrast to forests, now the walk may repeat vertices, so we erase its closed detours and flip the resulting simple directed path.  As in a forest, this decreases the outdegree of the starting vertex by one, increases that of the final vertex by one, and leaves all internal outdegrees unchanged.  Deletions require no repair.

The algorithm needs no auxiliary structure beyond the outgoing lists.  For constant $\Delta$, each step consists only of sampling from a constant-size list, and the total work is linear in the length of the sampled walk, repetitions included.  
All the difficulty is pushed into one combinatorial question: \EMPH{How long can the walk remain among vertices of outdegree at least $\Delta$}?

Our main result provides a \EMPH{tight} answer to this   question for outerplanar graphs, thus we resolve \Cref{q:intro-worst-case} for outerplanar graphs, albeit with
randomized update times rather than deterministic guarantees.
We note that, analogously to forests, there always (with probability 1) \EMPH{exists} an $O(\log n)$-length path that reaches a low-outdegree vertex.

\begin{theorem}[Dynamic outerplanar orientation]
\label{thm:intro-outerplanar}
Starting from the empty graph, the algorithm above maintains maximum
outdegree $4$ under fully dynamic updates that preserve outerplanarity.
Against an adaptive adversary, conditioned on the complete preceding
history, every update has expected update time and expected recourse
$O(\log n)$.  Moreover, throughout every polynomial-length execution,
all update times and recourse bounds are $O(\log n)$ simultaneously
w.h.p.
\end{theorem}

The hitting-time statement underlying
\Cref{thm:intro-outerplanar} is itself tight:
\vspace{-5pt}
\begin{restatable}[Tight outerplanar threshold]{theorem}{ThmTightOuter}
\label{thm:intro-tight-threshold}
In every orientation of every $n$-vertex outerplanar graph, from every
starting vertex, the uniform directed random walk reaches a vertex of
outdegree below $4$ after $O(\log n)$ steps in expectation and after
$O(\log(n/\delta))$ steps with probability at least $1-\delta$.

Conversely, for threshold $3$, there are $n$-vertex outerplanar
orientations and starting vertices for which the walk requires
$n^{\Omega(1)}$ steps in expectation and exceeds $n^{\Omega(1)}$ steps
with high probability.
For threshold $2$, the hitting time can be
$2^{\Omega(n)}$ in expectation and can exceed $2^{\Omega(n)}$ with
probability $1-2^{-\Omega(n)}$.  Both lower bound behaviors can arise
immediately after a single insertion.
\end{restatable}
\vspace{-5pt}

Interestingly, the tight outdegree $4$ threshold for this paradigm  is also twice the largest possible arboricity of an outerplanar graph, paralleling outdegree $2$ for forests. It also improves the outdegree bound of $(1+\eps)\alpha+2$ of the slower algorithms in
~\cite{CR22,CCHHQRS24} (with worst-case update time
$O(\eps^{-6} \log^3 n)$ for constant arboricity graphs). Indeed, for outerplanar graphs, $(1+\eps)\alpha+2$ gives outdegree 5 rather than 4; however, our algorithm is randomized and not deterministic as ~\cite{CR22,CCHHQRS24}, and it does not extend to arbitrary graphs of bounded arboricity.

Our upper bound analysis has two layers, both concerning the
same repair rule.  Layer I uses a simple path-counting argument
to establish the desired logarithmic bounds with threshold $48$,
cleanly exposing why the walk is short.  Layer II uses a more complex analysis to lower the threshold to the tight value $4$.
Thus only the analysis changes; the algorithm is identical in the two layers.

\vspace{-5pt}
\paragraph{Beyond outerplanar graphs.}
Since every outerplanar graph is $K_{2,3}$-minor-free, the family of $K_{2,3}$-minor-free graphs is a strict superclass of the outerplanar graphs: the additional $K_4$-minor exclusion that characterizes outerplanarity is not needed for the Layer-I analysis (for outdegree 48). Indeed, Layer I of our upper-bound analysis uses outerplanarity only through a bound on the number of simple paths of a given length between two fixed vertices. By extending this bound to $K_{2,t}$-minor-free graphs, we obtain the following result.
\vspace{-5pt}
\begin{restatable}[$K_{2,t}$-minor-free graphs]{theorem}{ThmMinorFree}
\label{thm:intro-k2t}
For every fixed $t\ge2$, fully dynamic $K_{2,t}$-minor-free graphs can be
maintained with maximum outdegree $48(t-1)$.  Against an adaptive
adversary, conditioned on the complete preceding history, every update
has expected update time and expected recourse $O(\log n)$.  Moreover,
throughout every polynomial-length execution, all update times and
recourse bounds are $O(\log n)$ simultaneously w.h.p.
\end{restatable}
\vspace{-3pt}
\noindent {\bf Remark.} We did not try to optimize the dependence on $t$ in the outdegree. 

\vspace{-8pt}

\paragraph{Dynamic distributed networks.}
The random walk repair procedure also admits an essentially lossless implementation
in the \EMPH{local-wakeup CONGEST model}~\cite{PPS16,CHK16,KS18,AOSS18,ALS22}.
Here the communication graph starts with no edges and undergoes one edge update per step; only the two updated endpoints initially wake up, and they must
restore the orientation using few synchronous rounds and few $O(\log n)$-bit messages.  We give the formal model in
\Cref{sec:dynamic-distributed}.

The entire repair procedure is carried by one \EMPH{token}.  It follows the sampled walk,
maintains its loop-erased path using local marks and predecessor pointers,
and finally backtracks along this path while flipping its edges.  A sampled
walk of length $L$, including repetitions, therefore costs $O(L)$ rounds
and $O(L)$ messages, without any global coordination.
\vspace{-5pt}
\begin{restatable}[Dynamic distributed orientation]{corollary}{CoroDistributed}
\label{cor:intro-distributed}
Under fully dynamic edge updates that preserve outerplanarity, maximum
outdegree $4$ can be maintained in the local-wakeup CONGEST model so that,
for every update and conditioned on the complete preceding history, the
expected round and message complexities are $O(\log n)$.  Throughout every
polynomial-length execution, all updates use $O(\log n)$ rounds and
messages simultaneously w.h.p.  For every fixed $t\ge2$, the analogous
result holds under updates that preserve $K_{2,t}$-minor-freeness, with
maximum outdegree $48(t-1)$.  These guarantees hold against an adaptive
adversary.
\end{restatable}
\vspace{-5pt}

This direct implementation is significant because many known centralized
orientation algorithms coordinate nonlocal repairs through global data
structures (e.g., a queue in the BF algorithm), so it is problematic to implement them efficiently in distributed networks.  The relevant
distributed orientation results~\cite{PPS16,KS18} focus primarily on
amortized guarantees and provide no comparable non-amortized message
bounds. Refer to
\Cref{sec:dynamic-distributed} for more details.

\vspace{-5pt}
\paragraph{Limitations.}
We show that the exponential hitting time of \Cref{thm:intro-tight-threshold}
for the threshold-$2$   walk stems from a more basic obstruction: even complete knowledge incurs $\Omega(n)$ \EMPH{amortized} recourse.

\vspace{-5pt}
\begin{restatable}[Outdegree $2$ requires $\Omega(n)$ recourse]{theorem}{ThmRecourse}
\label{thm:intro-outdegree-two}
For sufficiently large $n$, there is a sequence of $\Theta(n)$
updates, starting with no edges and preserving outerplanarity, such
that every sequence of orientations of  outdegree $2$ incurs
$\Omega(n^2)$ flips.  Thus, even
offline, the amortized recourse is
$\Omega(n)$.
\end{restatable}
\vspace{-5pt}
This complements the lower bound of~\cite{CHRT22}, which shows that maintaining a $3$-orientation in {\em planar} graphs requires $\Omega(n)$ amortized flips even offline~\cite{CHRT22}; Brodal and Fagerberg had earlier proved an analogous obstruction without the planarity restriction~\cite{BF99}.  Our construction lowers the outdegree to $2$ while restricting the graphs further to outerplanar graphs.

Next, we show that the random walk paradigm does not extend well to all series-parallel graphs.  These graphs are planar and have treewidth at most $2$, but they can contain $K_{2,m}$ as a minor for arbitrarily large $m$, creating many competing walks between the same two terminals.
\vspace{-5pt}
\begin{restatable}[Failure in series-parallel graphs]{theorem}{ThmSeriesParallel}
\label{thm:intro-series-parallel}
For every fixed $\Delta \ge3$, there are oriented $n$-vertex series-parallel
graphs and starting vertices from which the uniform directed walk needs
$n^{\Omega_\Delta(1)}$ steps, both in expectation and with high probability, to
reach outdegree below $\Delta$.  
\end{restatable}
\vspace{-5pt}

\vspace{-6pt}
\refstepcounter{subsection}\label{sec:intro-related}
\paragraph{1.4~ Further related work in a nutshell.~}

Low-outdegree orientations and closely related structures, such as \emph{forest}, \emph{pseudoforest}, and 
\emph{core decompositions}, 
have been studied extensively in the dynamic setting, as well as in various other settings and computational models, including the static sequential setting, the LOCAL and CONGEST models of distributed computing, the parallel and parallel batch-dynamic models; see ~\cite{MatulaBeck83,chiba1985arboricity,Kowalik06,BE08,PPS16,GS17,BHNT15,KS18,LSZ24,LSYDS22,DLRSSY22,GK25,BBDFGH26}, and the references therein.

Given the breadth of models in which edge orientation has been studied, the local, nearly structure-free nature of the random walk update rule is particularly appealing, as it seems \EMPH{amenable to extensions to additional computational models} beyond the dynamic settings studied in this work.

\vspace{-5pt}
\paragraph{Organization.}
\Cref{sec:outerplanar-layer1} presents the random walk repair rule
and proves the required bounds first at threshold $48$.
\Cref{sec:outerplanar-layer2} sharpens the analysis to
the tight threshold $4$, completing the proof of
\Cref{thm:intro-outerplanar}.
\Cref{sec:k2t-extension} extends Layer I to
$K_{2,t}$-minor-free graphs, and
\Cref{sec:dynamic-distributed} gives the dynamic distributed
implementation.
\Cref{sec:outerplanar-three-lower,sec:sp-walk-lower} establish the
threshold-$3$ and threshold-$2$ walk lower bounds and the
series-parallel obstruction, while
\Cref{sec:outdegree-two-lower-bound} proves the offline recourse lower
bound for outdegree $2$.
We conclude with open questions in \Cref{sec:conclusion}.
\vspace{-5pt}

\section{The Outerplanar Upper Bound:
Layer I---Threshold \texorpdfstring{$48$}{48}}
\label{sec:outerplanar-layer1}

We begin with a simple analysis at outdegree threshold $48$.  The repair rule is stated below for an arbitrary fixed threshold $\Delta$, because
Layer II uses exactly the same rule with $\Delta=4$; the two layers
differ only in their hitting-time analyses.

\subsection{The Random Walk Repair Rule}
\label{sec:randomwalk}

Fix an integer threshold $\Delta$.  Suppose that the current orientation
has maximum outdegree at most $\Delta$, and insert a new edge.

\begin{enumerate}[itemsep=0.4pt,parsep=0pt,topsep=2pt,partopsep=0pt]
  \item Name the endpoints $s,v$ and orient the new edge as $s\to v$.
  If $\outdeg(s)\le\Delta$, stop.

  \item Otherwise, $\outdeg(s)=\Delta+1$.  Starting at $s$, run the
  uniform directed random walk until it first reaches a vertex $z$ of
  outdegree less than $\Delta$.  Record the entire walk; do not mark
  vertices or forbid repetitions.

  \item Extract from the recorded walk its retained simple directed
  path from $s$ to $z$, as defined in
  \Cref{lem:decomposition}.

  \item Flip every edge of the retained path.
\end{enumerate}

Flipping the retained path restores the outdegree bound.  Its first
vertex $s$ loses one outgoing edge, its last vertex $z$ gains one, and
every internal vertex loses one outgoing path edge and gains another.
Thus $s$ returns from outdegree $\Delta+1$ to $\Delta$, while $z$,
which had outdegree less than $\Delta$, ends with outdegree at most
$\Delta$.  Deletions require no repair.

\noindent\textbf{Remark.}
We work in the randomized word-RAM model with $O(\log n)$-bit words and
constant-time sampling of a uniform integer from a constant-size range.
For fixed $\Delta$, during a repair every outdegree is at most
$\Delta+1$; hence each vertex's outgoing edges can be stored in a
fixed-size indexable array.  Thus sampling an outgoing edge and updating
an outgoing list take $O(1)$ time, and the total work is linear in the
length of the sampled walk, including repetitions.

\subsection{Removing Repetitions}
\label{sec:walk-decomposition}

A \emph{closed walk at $x$} is a walk that starts and ends at $x$; it may
have length zero.  The following elementary decomposition is the reason
that repeated vertices do not cause serious trouble.

\begin{lemma}[Walk decomposition]
\label{lem:decomposition}
Every directed walk $W=(x_0,x_1,\cdots, x_m)$ can be written as
\[
  W=C_0,e_0,C_1,e_1,\ldots,e_{k-1},C_k,
\]
where
\[
  P=(v_0,v_1,\ldots,v_k),
  \qquad v_0=x_0,\quad v_k=x_m,
\]
is a simple directed path, $e_i=v_i\to v_{i+1}$, and $C_i$ is a closed
walk at $v_i$.
Every vertex occurrence of $W$ belongs to exactly one $C_j$ of
these pieces. Every edge occurrence either lies in exactly one $C_j$ or is one of the retained edges $e_j$.  
The decomposition can be found in time linear in the length
of $W$.
\end{lemma}

\begin{proof}
First scan $W$ once and record the last position at which each visited
vertex occurs. 
Start at position $r=0$.  At each stage, let
$v_j=x_r$, and let $i$ be the last position at which $v_j$ occurs in $W$.
Set  $C_j=(x_r,x_{r+1},\ldots,x_i)$.
This is a closed walk at $v_j$, as $x_r=x_i=v_j$.
If $i=m$, stop.  Otherwise, retain the
next edge $e_j=x_i\to x_{i+1}$,
set $r=i+1$, and repeat.

A retained vertex cannot occur later in the walk, because we retain the
edge following its last occurrence.  Hence all retained vertices are
distinct, and the retained edges form a simple directed path from
$x_0$ to $x_m$.

The closed pieces are consecutive and partition the vertex occurrences
of $W$.  Together with the retained edges, they also partition its edge
occurrences.  After the initial scan, the construction takes time
linear in the length of $W$.
\end{proof}
Informally, $C_i$ contains all the wandering done at $v_i$ before the walk
leaves $v_i$ for the last time.  We call $P$ the \emph{retained path} of
$W$.

\subsection{The Walk Bound at Threshold \texorpdfstring{$48$}{48}}
\label{sec:walk-48}

Let $G$ be a simple outerplanar graph whose edges have been oriented.
Call a vertex \emph{high} if its outdegree is at least $48$, and
\emph{low} otherwise.  Starting at a vertex $s$, run the uniform
directed random walk until it first reaches a low vertex.  Let $\tau$
be the number of edges traversed before the walk stops; if $s$ is
already low, then $\tau=0$.

\begin{theorem}[Outerplanar walk with threshold $48$]
\label{thm:walk}
For every orientation of every $n$-vertex outerplanar graph, every starting
vertex $s$, and every integer $t\ge0$,
\[
  \Prob(\tau>t)\le 3n\,2^{-t}.
\]
Consequently, $\E[\tau]\le \left\lceil\log_2(3n)\right\rceil+2$.
Moreover, for every $0<\delta<1$,
\[
  \Prob\!\left(
    \tau>\left\lceil\log_2\frac{3n}{\delta}\right\rceil
  \right)
  \le\delta.
\]
In particular, $\tau=O(\log n)$ both in expectation and with high
probability.
\end{theorem}

\noindent\textbf{Remark.}
The value $48$ is not optimized; we choose it only to simplify the counting
argument.

We use the following fixed-endpoint bound of Matolcsi and Nagy
\cite[Section~3.1, proof of Theorem~1.8]{MN22}.  They let $f(k)$
denote the maximum number of simple paths with exactly $k$ edges between
two fixed vertices of an outerplanar graph, and prove that $f(k)\le 4^k$.
Their paths are undirected.  Since forgetting the edge directions maps
every directed simple $a$--$b$ path to a distinct undirected simple
$a$--$b$ path, the same bound applies to directed simple paths in an
orientation.

\begin{lemma}[Outerplanar path count]
\label{lem:path-count}
For fixed vertices $a,b$ of an outerplanar graph, the number of directed
simple $a$--$b$ paths with exactly $k$ edges in any orientation of the
graph is at most $4^k$.
\end{lemma}

A corresponding bound for $K_{2,t}$-minor-free graphs is proved in
\Cref{sec:k2t-extension}.

\subsection{Counting high walks}

Call a directed walk \emph{high} if every vertex on the walk is high.
For every directed edge $u\to v$ whose tail $u$ is high, define
\[
  \wt(u\to v):=\frac{2}{\outdeg(u)}.
\]
The weight of a high walk is the product of the weights of its edge
occurrences.  The empty walk has weight $1$.

\begin{observation}
\label{ob:basic}
These weights have three elementary properties.
\begin{enumerate} [itemsep=0.4pt,parsep=0pt,topsep=2pt,partopsep=0pt]
  \item Every edge of a high walk has weight at most
  $\frac{2}{48}=\frac1{24}$,
  as its tail has outdegree at least $48$.

  \item At every high vertex $u$, the weights of all outgoing edges add
  up to exactly $2$:
  \[
    \sum_{v:u\to v}\wt(u\to v)
    =
    \outdeg(u)\cdot\frac{2}{\outdeg(u)}
    =
    2.
  \]

  \item If $W=(x_0,x_1,\ldots,x_t)$
    is a particular high walk of length $t$, then the probability that the
  random walk follows $W$, starting from $x_0$, is
  \[
    \prod_{i=0}^{t-1}\frac{1}{\outdeg(x_i)}
    =
    2^{-t}\wt(W).
  \]
\end{enumerate}
Consequently, the total weight of all high walks of length $t$ starting
at $s$ is exactly
$2^t\Prob(\tau>t)$.
\end{observation}

We bound the total weight of a walk by grouping the walks according to their
retained paths.  Consider a high walk $W$ of length $t$ whose retained path is
  $P=(v_0,v_1,\ldots,v_k)$.
By Lemma~\ref{lem:decomposition}, the walk has the form
  $C_0,e_0,C_1,e_1,\ldots,e_{k-1},C_k$,
where $e_i=v_i\to v_{i+1}$ and $C_i$ is a high closed walk at $v_i$.
Because these pieces partition the edge occurrences of the original walk,
its weight is
\[
  \wt(W)
  =
  \left(\prod_{i=0}^{k-1}\wt(e_i)\right)
  \left(\prod_{i=0}^{k}\wt(C_i)\right).
\]

Suppose that the starting vertex $v_0=s$ is fixed.  There are at most
$n4^k$ possible retained paths with $k$ edges: there are at most $n$
choices for the last vertex $v_k$, and for each such choice
Lemma~\ref{lem:path-count} gives at most $4^k$ simple $s$--$v_k$ paths.
Every retained edge has a high tail and hence weight at most $1/24$.
Therefore, the $k$ retained edges of any fixed path contribute at most
$\left(\frac1{24}\right)^k$ to the weight.  It remains to account for
the $k+1$ closed pieces $C_0,C_1,\ldots,C_k$.

The next lemma (\Cref{lem:closed}) shows that, at any fixed high vertex, the total weight of
all possible high closed walks, subject to any fixed finite length bound,
is less than $2$.  Applying this with length bound $t$, and allowing the
closed pieces to be chosen independently, bounds their total contribution
by less than $2^{k+1}$. Of course, the independent choices only overcount: some choices may have total
length different from $t$, or may not produce the prescribed retained
path.

Consequently, once the closed-walk lemma is established, let
$\mathcal W_{t,k}$ denote the set of high walks of length $t$ starting at
$s$ whose retained path has $k$ edges.  The discussion above gives
\begin{equation}
\label{eq:fixed-k}
  \sum_{W\in\mathcal W_{t,k}}\wt(W)
  <
  n4^k
  \left(\frac1{24}\right)^k
  2^{k+1}
  =
  2n\left(\frac13\right)^k.
\end{equation}
The ratio $1/3$ has a simple meaning.  Increasing $k$ by one multiplies
the path-counting bound by $4$, contributes one more retained edge of
weight at most $1/24$, and creates one more closed-piece position of total
weight less than $2$.  Thus
  $4\cdot\frac1{24}\cdot2=\frac13$.
The leading factor $2$ in \eqref{eq:fixed-k} accounts for the first of the
$k+1$ closed-piece positions.  We now prove the bound of $2$ for one such
position.

\begin{lemma}[Closed walks cost less than two]
\label{lem:closed}
Fix an integer $L\ge0$.  For every high vertex $x$, the total weight of all
high closed walks at $x$ having at most $L$ edges is less than $2$.
\end{lemma}

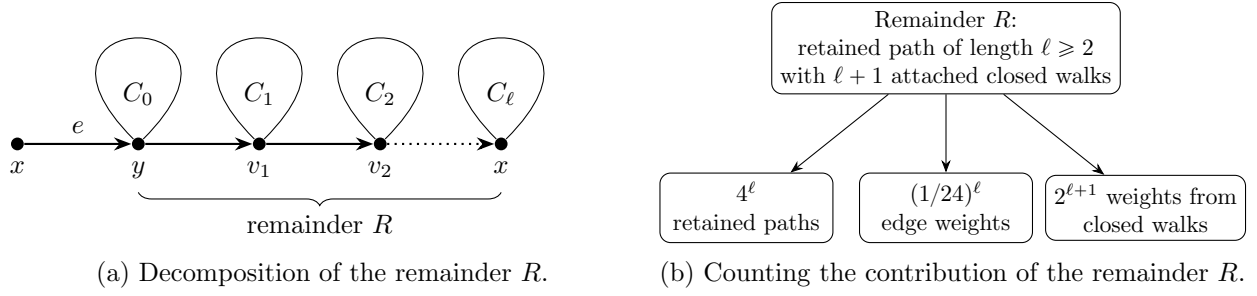
\begin{figure}[t]

\begin{subfigure}[t]{0.52\textwidth}

\begin{tikzpicture}[scale=0.8,
>=Stealth,
font=\small,
vertex/.style={
circle,
fill=black,
inner sep=1.7pt
}
]

\node[vertex,label=below:$x$] (x1) at (0,0) {};

\node[vertex,label=below:$y$] (y) at (2,0) {};

\node[vertex,label=below:$v_1$] (v1) at (4,0) {};

\node[vertex,label=below:$v_2$] (v2) at (6,0) {};

\node[vertex,label=below:$x$] (x2) at (8,0) {};

\draw[->,thick] (x1)--node[above]{$e$}(y);
\draw[->,thick] (y)--(v1);
\draw[->,thick] (v1)--(v2);
\draw[->,thick, dotted] (v2)--(x2);



\node at (2,0.82) {$C_0$};
\draw (y) to[out=45,in=135,looseness=50] (y);

\node at (4,0.82) {$C_1$};
\draw (v1) to[out=45,in=135,looseness=50] (v1);

\node at (6,0.82) {$C_2$};
\draw (v2) to[out=45,in=135,looseness=50] (v2);

\node at (8,0.82) {$C_\ell$};
\draw (x2) to[out=45,in=135,looseness=50] (x2);

\draw[
decorate,
decoration={brace,mirror,amplitude=5pt}
]
(2,-0.8)--(8,-0.8)
node[midway,yshift=-12pt]
{remainder $R$};

\end{tikzpicture}

\caption{Decomposition of the remainder $R$.}
\label{fig:decomposition}

\end{subfigure}
\hfill
\begin{subfigure}[t]{0.47\textwidth}
\centering

\begin{tikzpicture}[
scale=0.88,
transform shape,
>=Stealth,
font=\small,
box/.style={
draw,
rounded corners,
minimum width=2.6cm,
minimum height=9mm,
align=center
}
]

\node[box] (R) at (0,3)
{Remainder $R$:\\
retained path of length $\ell \ge 2$\\
with $\ell+1$ attached closed walks};

\node[box] (A) at (-3,0.6)
{$4^\ell$ \\
retained paths};

\node[box] (B) at (0,0.6)
{$(1/24)^\ell$\\
edge weights};

\node[box] (C) at (3,0.6)
{$2^{\ell+1}$ weights from\\
closed walks};

\draw[->] (R)--(A);
\draw[->] (R)--(B);
\draw[->] (R)--(C);

\end{tikzpicture}

\caption{Counting the contribution of the remainder $R$.}
\label{fig:recursive}

\end{subfigure}

\caption{
Illustration of the proof of Lemma~\ref{lem:closed}.
A non-empty high closed walk first traverses an edge
$x\to y$.
The remainder $R$ is decomposed according to Lemma~\ref{lem:decomposition}.
}

\label{fig:closed}

\end{figure}

\begin{proof}
For a high vertex $x$, let $S_L(x)$ denote the total weight of all high
closed walks from $x$ to itself having at most $L$ edges.  We prove by
induction on $L$, simultaneously for every high vertex $x$, that
$S_L(x) < 2$.
For $L=0$, the empty walk is the only such walk, so $S_0(x)=1 < 2$.

Now let $L\ge1$, and assume that $S_{L-1}(z) <2$, 
for every high vertex $z$.  Fix a high vertex $x$.  The empty walk
contributes $1$ to $S_L(x)$.  We next bound the contribution of the nonempty
walks.

Every nonempty high closed walk from $x$ begins with an edge  $e=x\to y$, where $y$ is high.  After this first edge, what remains is a high walk
$R$ from $y$ back to $x$ having at most $L-1$ edges. See Figure~\ref{fig:closed} for an illustration.
We first fix $e$
and bound the total weight of all possible remainders $R$.  The weight
of $e$ is not included yet; we will sum over the possible first edges
afterward.

Apply Lemma~\ref{lem:decomposition} to $R$.  Suppose that its retained
path has $\ell$ edges.  Since $e=x\to y$ and the graph is simple, the
retained $y$--$x$ path has $\ell\ge2$ edges.  For each $\ell$,
Lemma~\ref{lem:path-count} gives at most $4^\ell$ possible retained
$y$--$x$ paths.
For any such retained path, its $\ell$ edges contribute at most  $\left(\frac1{24}\right)^\ell$ to the weight.

The decomposition of $R$ also contains $\ell+1$ high closed pieces.
Each piece has at most $L-1$ edges, so by the induction hypothesis the
total weight of its possible choices is less than $2$.  Allowing the
$\ell+1$ pieces to be chosen independently therefore gives an upper
bound of $2^{\ell+1}$
for their total contribution to the weight; of course, this may count choices that do not
reconstruct a valid remainder $R$, but such over-counting only makes the estimate
larger.

Thus, for the fixed first edge $e$, the total weight of all possible
remainders is at most
\begin{align*}
  \sum_{\ell=2}^{L-1}
    4^\ell
    \left(\frac1{24}\right)^\ell
    2^{\ell+1}
  ~<~
    2\sum_{\ell\ge2}\left(\frac13\right)^\ell
  ~=~ \frac13.
\end{align*}

We now include the first edge.  Only edges $x\to y$ with $y$ high can
occur, and the total weight of these edges is at most the total weight
of all outgoing edges of $x$:
\[
  \sum_{\substack{x\to y\\ y\text{ high}}}\wt(x\to y)
  \le
  \sum_{x\to y}\wt(x\to y)
  =
  2,
\]
where the equality follows from \Cref{ob:basic}(2).
Consequently, all nonempty high closed walks at $x$ have total weight
at most
\[
  \frac13
  \sum_{\substack{x\to y\\ y\text{ high}}}\wt(x\to y)
  \le
  \frac13\cdot 2
  =
  \frac23.
\]
Adding the empty walk gives
$S_L(x)
  \le
  1+\frac23
  =
  \frac53 < 2$,
which completes the induction.
\end{proof}
We can now count all high prefixes of the random walk.

\begin{proof}[Proof of Theorem~\ref{thm:walk}]
If $s$ is low, then $\tau=0$.  Assume that $s$ is high, and fix $t\ge0$.
The sets $\mathcal W_{t,0},\ldots,\mathcal W_{t,t}$ partition the high
walks of length $t$ starting at $s$.  These walks are exactly the possible
length-$t$ prefixes on the event $\tau>t$.  By
\Cref{ob:basic}(3) and \eqref{eq:fixed-k},
\begin{align*}
  2^t\Prob(\tau>t)
  &=
  \sum_{k=0}^{t}\sum_{W\in\mathcal W_{t,k}}\wt(W) 
  ~<~
  \sum_{k=0}^{t}2n\left(\frac13\right)^k 
  ~<~
  2n\sum_{k\ge0}\left(\frac13\right)^k
  =3n.
\end{align*}
Therefore $\Prob(\tau>t)\le3n2^{-t}$.
Let $T=\lceil\log_2(3n)\rceil$.  For every $j\ge0$, the tail bound gives
\[
  \Prob(\tau>T+j)
  \le 3n\,2^{-(T+j)}
  \le 2^{-j}.
\]
Hence
\begin{align*}
  \E[\tau]
  &=
  \sum_{r=0}^{T-1}\Prob(\tau>r)
  +
  \sum_{j\ge0}\Prob(\tau>T+j) 
~\le~
  T+\sum_{j\ge0}2^{-j}
  =T+2.
\end{align*}
Finally, taking
$t=\lceil\log_2(3n/\delta)\rceil$ in the tail bound gives
$\Prob(\tau>t)\le\delta$. The theorem follows.
\end{proof}

\subsection{Dynamic Consequence at Threshold \texorpdfstring{$48$}{48}}

\begin{corollary}[Dynamic orientation with threshold $48$]
\label{cor:dynamic}
Starting from the empty graph, a simple outerplanar graph can be maintained
under insertions and deletions that preserve outerplanarity so that every
outdegree is at most $48$.
For every update, conditioned  on the complete
state before that update, the expected update time and recourse are
$O(\log n)$.  Moreover, for every $0<\delta<1$, both are
$O(\log(n/\delta))$ with probability at least $1-\delta$. Thus over every polynomial-length update sequence,
all update times and recourse bounds are $O(\log n)$
simultaneously with high probability.
\end{corollary}

\begin{proof}
A deletion needs no repair.  On an insertion, run the repair procedure from \Cref{sec:randomwalk}. If repair is needed, only the starting vertex $s$ has
outdegree $49$. By \Cref{thm:walk}, the sampled walk reaches a low vertex $z$ after
$O(\log n)$ steps in expectation and $O(\log(n/\delta))$ steps with
probability at least $1-\delta$.  By Lemma~\ref{lem:decomposition}, its retained path is no longer than the sampled walk. Flipping this path restores maximum outdegree $48$, as
explained in \Cref{sec:randomwalk}, and the number of flipped edges is at
most the walk length.

Recording the walk, extracting its retained path, and flipping the path
take time linear in the sampled length; by the representation described in
\Cref{sec:randomwalk}, each individual walk step and edge flip takes
constant time.  The claimed per-update bounds follow.  As
\Cref{thm:walk} holds for every current orientation, these bounds remain
valid after conditioning on the complete preceding history.

Finally, consider at most $n^c$ updates and fix any $\gamma>0$.  Apply the
per-update tail bound with $\delta=n^{-(c+\gamma)}$.  Every update then
takes $O(\log n)$ time and flips $O(\log n)$ edges except with
conditional probability at most $\delta$.  A union bound shows that all
updates satisfy these bounds with probability at least $1-n^{-\gamma}$.
\end{proof}
This completes Layer I: the   repair rule has the desired
logarithmic update time and recourse guarantees at threshold $48$.
Layer II sharpens the analysis of the same rule to
threshold $4$.

\section{The Outerplanar Upper Bound:
Layer II---Threshold \texorpdfstring{$4$}{4}}
\label{sec:outerplanar-layer2}

We now sharpen the hitting-time analysis of the repair rule from
\Cref{sec:randomwalk}.  Unlike Layer I, the proof below does not use
fixed-endpoint path counting; instead, it combines a two-terminal
estimate for outerplanar regions with a recursion at a centroid face.
For the uniform directed random walk $(X_j)_{j\ge0}$ in an orientation of a simple graph, put $\tau_4:=\min\{j\ge0:\outdeg(X_j)<4\}$.
We give a more precise statement of~\Cref{thm:intro-outerplanar}.

\begin{theorem}[Outerplanar walk with threshold $4$]
\label{thm:threshold-four}

There are absolute constants $\zeta>1$, $C>0$, and $\kappa>0$ such that,
for every orientation of every $n$-vertex outerplanar graph and every
starting vertex $s$, $\E_s[\zeta^{\tau_4}]\le Cn^\kappa$.
Consequently, $\Prob_s(\tau_4>m)\le Cn^\kappa\zeta^{-m}$ (for $m\ge0$), $\E_s[\tau_4]=O(\log n)$, and
$\tau_4=O(\log(n/\delta))$ with probability at least $1-\delta$.
\end{theorem}

\subsection{Proof of Theorem~\ref{thm:threshold-four}}

We prove a slightly stronger statement.
Consider a finite substochastic
Markov chain on the vertices of an oriented simple outerplanar graph.
Assume that $P(u,v)>0$ only if $u\to v$ is an oriented edge and that $P(u,v)\le\frac14$
for every $u,v$.  Missing probability is sent to an absorbing cemetery
state.
We write $\tau$ for the absorption time and prove a polynomial
bound on $\E_v[\zeta^\tau]$.

For an outerplanar graph $G$, fix an embedding in the plane.
Assume for now that $N\ge4$, and augment $G$, on the
same vertex set, to a maximal outerplanar graph.
Thus, all faces of the augmented graph are triangles and the augmented graph is $2$-connected.
The added edges have probability $0$.
We do not assume that the augmented graph is oriented with bounded outdegree, only that $P(u,v)\leq 1/4$ for all pairs.

It is convenient to scale probabilities by $4$: write
$w(u,v):=4P(u,v)$.  Every edge then has weight at most $1$, and the total
outgoing weight at a vertex is at most $4$.
For a stopping time $T$ and an event $\mathcal E$ determined
by the walk up to time $T$, we call
$\E[\zeta^T\mathbf 1_{\mathcal E}]$
the {$\zeta$-weight} of $\mathcal E$.
Thus, traversing an edge of weight $w$ contributes the factor $\zeta w/4$.

\begin{definition}[$\zeta$-weight of a trajectory]
\label{def:traj}
A \emph{trajectory} of length $k$ is a finite sequence of states
$w = (x_0, x_1, \dots, x_k)$ such that $P(x_i, x_{i+1}) > 0$ for every $0 \le i < k$. Write $|w| := k$ for its length and
\[
  \Pr(w) \;:=\; \prod_{i=0}^{k-1} P(x_i, x_{i+1})
\]
for the probability that the walk, started at $x_0$, follows exactly this sequence of states for its first $k$ steps. The \emph{$\zeta$-weight} of $w$ is
\[
  \operatorname{wt}_\zeta(w) \;:=\; \zeta^{|w|} \cdot \Pr(w).
\]
\end{definition}
 
\begin{definition}[$\zeta$-weight of an event]
\label{def:event}
Let $T$ be a stopping time for the walk, and let $E$ be an event determined by the walk's trajectory up to time $T$. The \emph{$\zeta$-weight of $E$}, started from a state $u$, is
\[
  \operatorname{wt}_\zeta(E) \;:=\; \sum_{w \,\in\, E} \operatorname{wt}_\zeta(w)
  \;=\; \mathbb{E}_u\!\left[\zeta^{T} \, \mathbf{1}_E\right].
\]
That is, $\operatorname{wt}_\zeta(E)$ sums the $\zeta$-weights (Definition~\ref{def:traj}) of every trajectory belonging to $E$, or equivalently is the expectation of $\zeta^T$ restricted to the event $E$.
\end{definition}

Let us briefly describe the weak dual decomposition of outerplanar graphs that we will use in the proof.

\begin{definition}[Weak dual]
For an outerplanar graph $G$, the \emph{weak dual} of $G$ is graph $D(G)$ where every (inner) face of $G$ is a vertex and two faces have an edge iff they share an edge in $G$.
Because $G$ is outerplanar, $D(G)$ is a forest.
\end{definition}
For simplicity,
we assume that $G$ is a maximal outerplanar graph,
so that $G$ is $2$-connected and all faces are triangles,
and thus  $D(G)$ is a tree with maximum degree $3$.
See~\Cref{fig:weak-dual}.

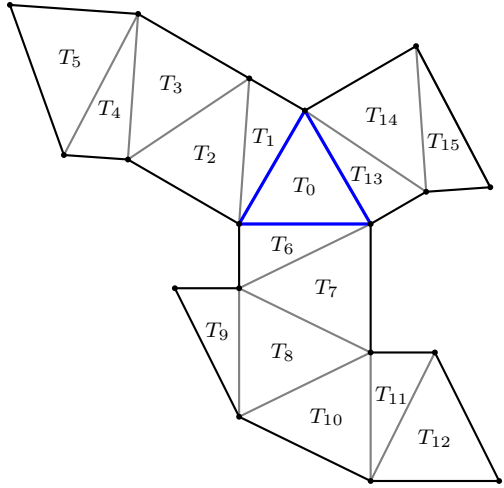
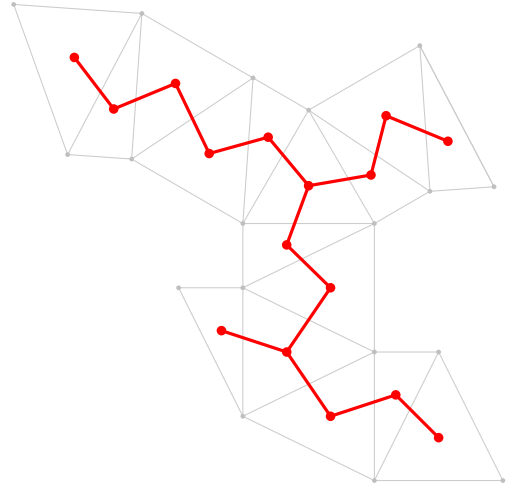
\begin{figure}

\begin{subfigure}[t]{0.6\textwidth}
\begin{tikzpicture}[
>=Stealth,
every node/.style={font=\scriptsize}
]

    \coordinate (P0)  at (0.000, 1.000);
    \coordinate (P1)  at (-0.870,-0.500);
    \coordinate (P2)  at (0.870,-0.500);
    \coordinate (B11) at (-0.735, 1.426);
    \coordinate (B12) at (-2.341, 0.353);
    \coordinate (B13) at (-2.206, 2.279);
    \coordinate (B24) at (-0.870,-1.350);
    \coordinate (B25) at (0.870,-2.200);
    \coordinate (B26) at (-0.870,-3.050);
    \coordinate (B27) at (0.870,-3.900);
    \coordinate (B38) at (1.605,-0.074);
    \coordinate (B39) at (1.471, 1.853);
    \coordinate (C110) at (-3.188, 0.412);
    \coordinate (C111) at (-3.902, 2.398);
    \coordinate (C212) at (1.720,-2.200);
    \coordinate (C213) at (2.570,-3.900);
    \coordinate (C314) at (-1.720,-1.350);
    \coordinate (C415) at (2.453,-0.014);

    \draw[thick, gray]
      (P0)  -- (P1)
      (P1)  -- (P2)
      (P0)  -- (P2)
      (P1)  -- (B11)
      (B11) -- (B12)
      (B12) -- (B13)
      (P2)  -- (B24)
      (B24) -- (B25)
      (B25) -- (B26)
      (B24) -- (B26)
      (B25) -- (B27)
      (P0)  -- (B38)
      (B38) -- (B39)
      (B13) -- (C110)
      (B27) -- (C212)
      (B39) -- (C415);

    \draw[thick]
      (P0)  -- (B11) -- (B13) -- (C111) -- (C110) -- (B12) -- (P1)
            -- (B24) -- (C314) -- (B26) -- (B27) -- (C213) -- (C212) -- (B25) -- (P2)
            -- (B38) -- (C415) -- (B39) -- cycle;

    \draw[very thick, blue] (P0) -- (P1) -- (P2) -- cycle;

    \foreach \v in {P0,P1,P2,B11,B12,B13,C110,C111,
                    B24,B25,B26,B27,C212,C213,C314,
                    B38,B39,C415} {
      \fill (\v) circle (1.1pt);
    }

    \node at (0.0,0.0)      {$T_{0}$};
    \node at (-0.535,0.642) {$T_{1}$};
    \node at (-1.315,0.426) {$T_{2}$};
    \node at (-1.761,1.353) {$T_{3}$};
    \node at (-0.29,-0.783) {$T_{6}$};
    \node at (0.29,-1.35)   {$T_{7}$};
    \node at (-0.29,-2.2)   {$T_{8}$};
    \node at (0.29,-3.05)   {$T_{10}$};
    \node at (0.825,0.142)  {$T_{13}$};
    \node at (1.025,0.926)  {$T_{14}$};
    \node at (-2.578,1.015) {$T_{4}$}; 
    \node at (-3.099,1.697) {$T_{5}$};
    \node at (1.153,-2.767) {$T_{11}$};
    \node at (1.72,-3.333)  {$T_{12}$};
    \node at (-1.153,-1.917){$T_{9}$};
    \node at (1.843,0.588)  {$T_{15}$};

\end{tikzpicture}
\caption{A triangulated (maximal) outerplanar graph.}
\end{subfigure}
\hfill
\begin{subfigure}[t]{0.5\textwidth}
\begin{tikzpicture}[
>=Stealth,
every node/.style={font=\scriptsize}
]
        
    \coordinate (P0)  at (0.000, 1.000);
    \coordinate (P1)  at (-0.870,-0.500);
    \coordinate (P2)  at (0.870,-0.500);
    \coordinate (B11) at (-0.735, 1.426);
    \coordinate (B12) at (-2.341, 0.353);
    \coordinate (B13) at (-2.206, 2.279);
    \coordinate (B24) at (-0.870,-1.350);
    \coordinate (B25) at (0.870,-2.200);
    \coordinate (B26) at (-0.870,-3.050);
    \coordinate (B27) at (0.870,-3.900);
    \coordinate (B38) at (1.605,-0.074);
    \coordinate (B39) at (1.471, 1.853);
    \coordinate (C110) at (-3.188, 0.412);
    \coordinate (C111) at (-3.902, 2.398);
    \coordinate (C212) at (1.720,-2.200);
    \coordinate (C213) at (2.570,-3.900);
    \coordinate (C314) at (-1.720,-1.350);
    \coordinate (C415) at (2.453,-0.014);

    \draw[thin, gray!40]
      (P0)  -- (P1)
      (P1)  -- (P2)
      (P0)  -- (P2)
      (P1)  -- (B11)
      (B11) -- (B12)
      (B12) -- (B13)
      (P2)  -- (B24)
      (B24) -- (B25)
      (B25) -- (B26)
      (B24) -- (B26)
      (B25) -- (B27)
      (P0)  -- (B38)
      (B38) -- (B39)
      (B13) -- (C110)
      (B27) -- (C212)
      (B39) -- (C415);
    \draw[thin, gray!40]
      (P0)  -- (B11) -- (B13) -- (C111) -- (C110) -- (B12) -- (P1)
            -- (B24) -- (C314) -- (B26) -- (B27) -- (C213) -- (C212) -- (B25) -- (P2)
            -- (B38) -- (C415) -- (B39) -- cycle;
    \foreach \v in {P0,P1,P2,B11,B12,B13,C110,C111,
                    B24,B25,B26,B27,C212,C213,C314,
                    B38,B39,C415} {
      \fill[gray!50] (\v) circle (0.9pt);
    }

    \coordinate (D0)  at (0.0,0.0);
    \coordinate (D1)  at (-0.535,0.642);
    \coordinate (D2)  at (-1.315,0.426);
    \coordinate (D3)  at (-1.761,1.353);
    \coordinate (D4)  at (-0.29,-0.783);
    \coordinate (D5)  at (0.29,-1.35);
    \coordinate (D6)  at (-0.29,-2.2);
    \coordinate (D7)  at (0.29,-3.05);
    \coordinate (D8)  at (0.825,0.142);
    \coordinate (D9)  at (1.025,0.926);
    \coordinate (D10) at (-2.578,1.015);
    \coordinate (D11) at (-3.099,1.697);
    \coordinate (D12) at (1.153,-2.767);
    \coordinate (D13) at (1.72,-3.333);
    \coordinate (D14) at (-1.153,-1.917);
    \coordinate (D15) at (1.843,0.588);

    \draw[very thick, red]
      (D0)  edge[-] (D1)
      (D0)  edge[-] (D4)
      (D0)  edge[-] (D8)
      (D1)  edge[-] (D2)
      (D2)  edge[-] (D3)
      (D3)  edge[-] (D10)
      (D4)  edge[-] (D5)
      (D5)  edge[-] (D6)
      (D6)  edge[-] (D7)
      (D6)  edge[-] (D14)
      (D7)  edge[-] (D12)
      (D8)  edge[-] (D9)
      (D9)  edge[-] (D15)
      (D10) edge[-] (D11)
      (D12) edge[-] (D13);

    \foreach \v in {D0,D1,D2,D3,D4,D5,D6,D7,D8,D9,D10,D11,D12,D13,D14,D15} {
      \fill[red] (\v) circle (1.8pt);
    }

    \end{tikzpicture}
    \caption{The weak dual decomposition.}
\end{subfigure}
    \caption{A triangulated outerplanar graph and its weak dual decomposition.}
    \label{fig:weak-dual}
\end{figure}

The proof has two steps.  First we show that, inside any outerplanar
region, the $\zeta$-weight of reaching its two boundary vertices is
bounded by a constant $h<1$.  We then cut the graph at a centroid
triangle and sum the resulting excursions as a geometric series.

\subsection{A two-terminal estimate.}
Taking an edge $ab$ in $G$, let $R:=R(a,b)$ be the
region on one side of $ab$, with terminals $a,b$.
That is, $R$ is the union of all faces in this one side of $ab$.
Starting from a terminal $u\in\{a,b\}$, count only trajectories whose first step enters
$R^\circ:=V(R)\setminus\{a,b\}$.
and stop when the walk next reaches
$a$, $b$, or the cemetery state. 
For $v\in\{a,b\}$, let $K_R(u,v)$ be the $\zeta$-weight mass of a walk from terminal $u$ exiting $R$ through terminal $v$.
This exit mass is bounded purely in terms of $u$'s own outgoing weight into $R^\circ$ --- not in terms of $|R^\circ|$ --- and consequently that any interior vertex reaches $\{a,b\}$ with $\zeta$-weighted mass at most $h=18/19<1$. The proof inducts on $|R^\circ|$, splitting each
region at one triangle into two smaller children and reconstructing the
parent's bound from theirs; the reconstruction reduces to numeric inequalities in $\theta,\sigma,\nu,\zeta$ that were
pre-selected given that each vertex's total outgoing weight
is $4$.

\begin{definition}[Special case: $K_R(u,v)$]
\label{def:KR}
Let $R$ be a region with two designated terminal vertices $a,b$, and let $u, v \in \{a,b\}$. Run the walk starting at $u$, and let
\[
  T \;:=\; \min\{\, j > 0 : X_j \in \{a,b\}\ \text{or the walk is absorbed} \,\}
\]
be the first time the walk exits $R$ through a terminal or is absorbed inside $R$. Let
\[
  E_v \;:=\; \{\, T < \infty \ \text{and the walk exits through } v \,\}
\]
be the event that the walk takes its first step into $R^{\circ}$ and exits $R^{\circ}$ through $v$ (rather than being absorbed). Define
\[
  K_R(u,v) \;:=\; \operatorname{wt}_\zeta(E_v) \;=\; \mathbb{E}_u\!\left[\zeta^{T}\,\mathbf{1}_{E_v}\right],
\]
the $\zeta$-weight (Definition~\ref{def:event}) of the event that the walk started at $u$ exits region $R$ through terminal $v$. Trajectories that are absorbed inside $R$ before reaching $\{a,b\}$ do not belong to $E_v$, and contribute $0$ to $K_R(u,v)$. We call $K_R(u,v)$ the ``regional kernel''.
\end{definition}

For $u = a$, the two quantities $K_R(a,b)$ and $K_R(a,a)$ are respectively the $\zeta$-weighted mass of trajectories from $a$ that exit through $b$, and the $\zeta$-weighted mass of trajectories from $a$ that loop back and exit through $a$ itself. Let $p_R(u):=\sum_{z\in R^\circ}w(u,z)$ be the total outgoing weight from $u$ into the region.
For the starting point $a$, write
\[
 x:=K_R(a,b),\qquad y:=K_R(a,a),\qquad p:=p_R(a),
\]
and set $\overline{p}:=\min\{p,1\}$.  We use the following parameters
\[
 \theta:=\frac9{10},\qquad \sigma:=\frac78,
 \qquad \nu:=\frac{17}{20},
 \qquad
 A(p):=\frac p5+\frac{\overline{p}}{50},
 \qquad B(p):=\frac{\nu p}{4}.
\]
Note that $\theta-\sigma=\sigma-\nu=1/40$ and $1/5 < \sigma/4 < 1/5 + 1/50$.
The term $p/5$ is the main budget and the small extra term in $A$ is a
one-time reserve used when the parent region contains a direct edge from
a terminal to its root.
We prove the two inequalities
\begin{equation}\label{eq:cone}
  x+\theta y\le A(p),\qquad y+\theta x\le B(p).
\end{equation}

We need two elementary consequences.  The first controls repeated returns
at the internal vertex of a region; the second pays for a possible direct
edge from a terminal to that vertex.  Adding the inequalities in
\eqref{eq:cone} gives
\begin{equation}\label{eq:phi}
 \theta(x+y)\le
 \Phi(p):=\frac{\theta}{1+\theta}\bigl(A(p)+B(p)\bigr)
 =\frac p5+\frac{72\overline{p}-35p}{7600}
 \le\frac p5+\frac{37}{7600}
 <\frac p5+\frac1{200}.
\end{equation}
The last inequality follows immediately by considering $p\le1$ and $p\ge1$.
Second, we find a lower bound for $A(p+e)$ for a direct edge of any weight $e\in[0,1]$.
Since $1>\sigma>\theta^2$, note that
\begin{align*}
    \sigma x+\theta y
    &=\frac{\sigma-\theta^2}{1-\theta^2}(x+\theta y) +\frac{\theta(1-\sigma)}{1-\theta^2}(y+\theta x)\\
    &\le\frac{\sigma-\theta^2}{1-\theta^2}A(p) +\frac{\theta(1-\sigma)}{1-\theta^2}B(p)
    =A(p)+\frac{1-\sigma}{1-\theta^2}\Big(\theta B(p)-A(p)\Big)
    \le A(p)-\frac{3\overline{p}}{160}.
\end{align*}
Substituting $A,B$ leaves a nonnegative slack $(35p-34\overline{p})/6080\geq 0$.
Next, note that: $\overline{p+e}-\overline{p}\geq (e-\overline{p})\cdot 15/16$,
considering  $p\leq 1-e$ and $p>1-e$.
This reserve of $3\overline{p}/160$ pays for the weight $e$.
Namely, we get
$A(p+e)-A(p)-{\sigma e}/{4}= -3e/160+(\overline{p+e}-\overline{p})/50 \ge -{3\overline{p}}/{160}$, which yields:
\begin{equation}\label{eq:reserve}
\sigma x+\theta y+\frac{\sigma e}{4}\le A(p+e)~.
\end{equation}

Lastly, we choose $\zeta$.
Consider  $q,r\in[0,4]$ and $f,g\in[0,1]$.
These variables represent the regional kernels of two child regions in the weak dual graph and the weights of the connecting edge, respectively.
From~\eqref{eq:phi}, note that $A(r)+\Phi(q)<\frac{1}{5}r+\frac{1}{50}+\frac{1}{5}q+\frac{1}{200}=\frac{1}{5}(q+r)+\frac{1}{40}$.
At $\zeta=1$, the relevant inequalities are: 
\begin{align*}
 \frac{g+\theta f}{4}+A(r)+\Phi(q)
 &<\frac{33+f+2g}{40}\le\theta
 &&(q+r+f+g\le4),\\
 \frac g4+A(r)+\Phi(q)
 &<\frac{33+2g}{40}\le\sigma
 &&(q+r+g\le4),\\
 \frac{\theta g}{4}+A(q)+\Phi(r)
 &<\frac{33+g}{40}\le\nu
 &&(q+r+g\le4).
\end{align*}
These inequalities are strict
and all have a uniform positive slack (at least $1/7600$ from \eqref{eq:phi}).
We may therefore fix $\zeta>1$,
sufficiently close to $1$ (namely, $\zeta\leq 1+1/50$), such that
\begin{align}
 \frac\zeta4(g+\theta f)+A(r)+\Phi(q)&\le\theta,
   &&q+r+f+g\le4,\label{eq:budget-ordinary}\\
 \frac{\zeta g}{4}+A(r)+\Phi(q)&\le\frac\sigma\zeta,
   &&q+r+g\le4,\label{eq:budget-U}\\
 \frac{\zeta\theta g}{4}+A(q)+\Phi(r)&\le\frac\nu\zeta,
   &&q+r+g\le4.\label{eq:budget-V}
\end{align}

Now we prove that for every terminal $a,b$, $K_R(a,b)$ and $K_R(a,a)$ satisfy \eqref{eq:cone}. Intuitively, the $\zeta$-weighted mass of trajectories from $a$ that exist throught $a$ or $b$ only depend on the total outgoing weight from $a$ into the region

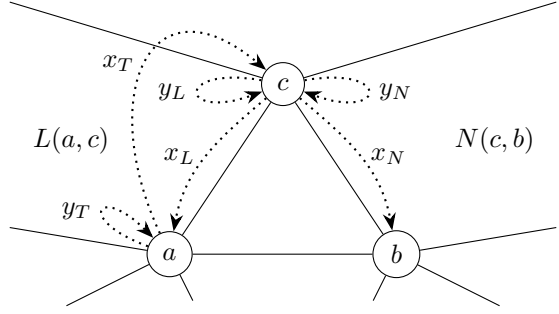
\begin{figure}[t]
\centering
\begin{tikzpicture}[
scale=1.5,
>=Stealth,
font=\small,
vertex/.style={
circle,
draw,
fill=white,
inner sep=2pt
}
]

\node[vertex] (a) at (0,0) {${\;}a{\;}$};
\node[vertex] (b) at (2,0) {${\;}b{\;}$};
\node[vertex] (c) at (1,3/2) {${\;}c{\;}$};
\node[] (a0) at (-1,-1/2) {};
\node[] (b0) at (3,-1/2) {};
\node[] (c0) at (7/4,-1/2) {};
\node[] (d0) at (1/4,-1/2) {};
\node[] (a1) at (-1-1/2,1/4) {};
\node[] (b1) at (3+1/2,1/4) {};
\node[] (c1) at (-1-1/2,2+1/4) {};
\node[] (c2) at (3+1/2,2+1/4) {};

\draw (a) edge[-] (b);
\draw (c) edge[-] (a) (a) edge[-] (a0) (a) edge[-] (d0);
\draw (c) edge[-] (b) (b) edge[-] (b0) (b) edge[-] (c0);
\draw (a) edge[-] (a1) (b) edge[-] (b1);
\draw (c) edge[-] (c1) (c) edge[-] (c2);

\node[] (L) at (0-3/4, 1) {$L(a,c){\quad}$};
\node[] (N) at (2+3/4, 1) {${\quad}N(c,b)$};

\draw[->,dotted, thick] (c) to[out=-140, in=80, looseness=1] node[left]{$x_{L}$} (a);
\draw[->,dotted, thick] (c) to[out=170, in=-160, looseness=20] node[left]{$y_{L}$} (c);
\draw[->,dotted, thick] (c) to[out=-40, in=100, looseness=1] node[right]{$x_{N}$} (b);
\draw[->,dotted, thick] (c) to[out=10, in=-20, looseness=20] node[right]{$y_{N}$} (c);
\draw[->,dotted, thick] (a) to[out=115, in=140, looseness=2] node[left]{$x_{T}$} (c);
\draw[->,dotted, thick] (a) to[out=160, in=135, looseness=20] node[left]{$y_{T}$} (a);

\end{tikzpicture}
\caption{Illustration of the proof of~\Cref{lem:region}.
The subpolygon $R(a,b)$ is the union of $L(a,c)$, $N(c,b)$, and the triangle $abc$.
}
\label{fig:regions}
\end{figure}

\begin{lemma}\label{lem:region}
For every region $R(a,b)$ of the triangulated outerplanar graph, $K_R(a,b)$ and $K_R(a,a)$ satisfy \eqref{eq:cone}. 
Moreover, from every vertex $v$ of $R^\circ$, the total $\zeta$-weighted mass of trajectories starting at $v$ and reaching terminals $a$ or $b$ before absorption is at most $$h:=\frac{2\theta}{1+\theta}=\frac{18}{19}~.$$
\end{lemma}

\begin{proof}
We prove both statements by induction on $|R^\circ|$.
The empty region is immediate.  Let $c$ be the third vertex of the
triangle incident with $ab$, and let $L(a,c)$ and $N(c,b)$ be the two
child regions.
See~\Cref{fig:regions}.
Write
\begin{align*}
    (x_T,y_T):=(K_L(a,c),K_L(a,a))~,&&
    (x_L,y_L):=(K_L(c,a),K_L(c,c))~,&&
    (x_N,y_N):=(K_N(c,b),K_N(c,c))~.
\end{align*}
Let $e,f,g\in[0,1]$ be the weights of $a\to c$, $c\to a$, and
$c\to b$, and put
\begin{align*}
    p:=p_L(a)~,&&
    q:=p_L(c)~,&&
    r:=p_N(c)~.
\end{align*}
Then, $q+r+f+g\le4$
and $e>0$ forces $f=0$.
Moreover, $p_R(a)=p+e$.
By induction, we have
\begin{align*}
x_L+\theta y_L \leq A(q)~, && 
y_L+\theta x_L \leq B(q)~, && 
x_N+\theta y_N \leq A(r)~, &&
y_N+\theta x_N \leq B(r)~.
\end{align*}
Thus,
$y_L+y_N\le B(q)+B(r)\le B(4)=17/20<1$. Repeated returns to $c$ give
\[
 H_a:=\frac{\zeta f/4+x_L}{1-y_L-y_N},\qquad
 H_b:=\frac{\zeta g/4+x_N}{1-y_L-y_N},
\]
Here $H_a$ is the total $\zeta$-weighted mass of a walk started at $c$
eventually exiting toward terminal $a$, accounting for the walk possibly
looping back to $c$ via $L$ or $N$ arbitrarily many times before
leaving: the numerator $\zeta f/4+x_L$ is the mass of
reaching $a$ without any return to $c$, and
dividing by $1-y_L-y_N$ sums the resulting geometric series over repeated returns $\sum_{i=0}^{\infty}(y_N+y_L)^i$, since each return resets the walk to $c$. $H_b$
is defined symmetrically for terminal $b$.

For convenience, put $U:=H_b+\theta H_a$ and $V:=H_a+\theta H_b$. Then
\begin{align*}
 (1-y_L-y_N)\cdot U+\theta(y_L+y_N)
 &=\frac\zeta4(g+\theta f)
   +(x_N+\theta y_N)+\theta(x_L+y_L)\\
 &\le\frac\zeta4(g+\theta f)+A(r)+\Phi(q)\le\theta,
\end{align*}
so $U\le\theta$.  The symmetric calculation gives $V\le\theta$.
Consequently,
\begin{equation}\label{eq:hitting-root}
 H_a+H_b=\frac{U+V}{1+\theta}\le h.
\end{equation}

\noindent
The regional kernel of $R(a,b)$ is thus:
\begin{align*}
    K_R(a,b)=(x_T+\zeta e/4)H_b~,&&
    K_R(a,a)=y_T+(x_T+\zeta e/4)H_a~.
\end{align*}
If $e=0$, then
\begin{align*}
 K_R(a,b)+\theta K_R(a,a)
 &= x_T\cdot U+\theta y_T
 \le\theta x_T+\theta y_T
 \le x_T+\theta y_T\le A(p),\\
 K_R(a,a)+\theta K_R(a,b)
 &=y_T+x_T\cdot V
 \le y_T+\theta x_T\le B(p).
\end{align*}
If $e>0$, then $f=0$.  Since $\sigma,\nu<\theta$, \ref{eq:budget-U} and \ref{eq:budget-V} give:
\begin{align*}
 (1-y_L-y_N)\cdot U+\frac\sigma\zeta(y_L+y_N)<
 \frac{\zeta g}{4}+x_N+\theta x_L+\theta(y_L+y_N)
 &\le\frac{\zeta g}{4}+A(r)+\Phi(q)\le\frac\sigma\zeta,\\
 (1-y_L-y_N)\cdot V+\frac\nu\zeta(y_L+y_N)<
 x_L+\theta\left(\frac{\zeta g}{4}+x_N\right)+\theta(y_L+y_N)
 &\le\frac{\zeta\theta g}{4}+A(q)+\Phi(r)\le\frac\nu\zeta.
\end{align*}
Thus $U\le\sigma/\zeta$ and $V\le\nu/\zeta$.
Using $\zeta>1$ and \eqref{eq:reserve} gives:
\begin{align*}
 K_R(a,b)+\theta K_R(a,a)
 &=\left(x_T+\frac{\zeta e}4\right)\cdot U+\theta y_T
 \le\frac\sigma\zeta x_T+\frac{\sigma e}{4}+\theta y_T
 \le\sigma x_T+\frac{\sigma e}{4}+\theta y_T
 \le A(p+e),\\
 K_R(a,a)+\theta K_R(a,b)
 &=y_T+\left(x_T+\frac{\zeta e}4\right)\cdot V
 \le y_T+\theta x_T+\frac{\nu e}{4}
 \le B(p)+\frac{\nu e}{4}=B(p+e).
\end{align*}

Finally, \eqref{eq:hitting-root} proves the terminal-hitting estimate from $c$.
From a vertex in a child region, first stop at the child's two terminals.
By induction, their total hitting mass is at most $h$;
from either terminal, the $\zeta$-weighted mass toward $\{a,b\}$ is at most $1$.
This proves the hitting estimate and completes the induction.
\end{proof}

\subsection{The centroid recursion.}
 
We prove by induction on $N$ that there is an absolute constant $C\ge 1$,
independent of $N$, such that
\[
  \mathbb E_v[\zeta^\tau] \;\le\; CN^\kappa \qquad\text{for every starting vertex } v,
\]
Here $\tau$ is the absorption time of the stochastic chain of Theorem~\ref{thm:threshold-four}. The proof bounds the
$\zeta$-weighted cost of one round trip from centroid face $S$ and back by a fixed ratio $\rho<1$, and sums the resulting geometric series against an inductive bound on the half-sized sub-problem.

\smallskip
\noindent Define the following parameters:
\begin{align*}
 \rho:=\frac{\zeta(1+h)}2=\frac{37\zeta}{38}<1~,&&
 \Lambda:=\frac\zeta{1-\rho}~,&&
 \kappa:=\log_2 (1+\Lambda)~.
\end{align*}

\smallskip
\begin{lemma}[Centroid]
\label{lem:centroid}
Every tree on $N\ge 1$ nodes has a node $S$ whose removal leaves every
component with at most $\lceil N/2\rceil$ nodes.
\end{lemma}

\smallskip
\noindent\textbf{Base case ($N\le 3$).} With at most $3$ transient states,
$\Pr(\tau>j)\le 2^{-j}$ for all $j$, so $\tau<\infty$ almost surely, and
\[
 \E_v[\zeta^\tau]
 =1+(\zeta-1)\sum_{j\ge0}\zeta^j\Prob_v(\tau>j)
 \le\frac\zeta{2-\zeta}.
\]
Since $\zeta<38/37<2$, this is a finite constant; fix $C$ to dominate it.
 
\smallskip
\noindent\textbf{Inductive step ($N\ge4$).} Assume the bound for every
instance with fewer than $N$ transient states. Let $S$ be a centroid face of
$D(G)$ (Lemma~\ref{lem:centroid}); removing $S$ leaves at most three
attached regions, each with at most $N/2$ vertices, to which the inductive
hypothesis applies.
 
Fix a vertex of $S$ and call an \emph{excursion} the walk's trajectory from
the moment it leaves $S$ until it either (a) is absorbed without returning to
$S$, or (b) returns to a vertex of $S$. Let $d\le2$ be the total weight of the
direct edges from the starting vertex of $S$ to $S$'s other two vertices; the
remaining weight $4-d$ enters the attached regions.
 
\begin{lemma}[One excursion]
\label{lem:excursion}
The $\zeta$-weighted mass of return (case (b)) on one excursion is at most
\[
  \rho \;:=\; \frac{\zeta}{4}(2+2h) \;=\; \frac{\zeta(1+h)}{2} \;<\;1.
\]
\end{lemma}
 
\begin{proof}
With weight $d$ the walk exits directly to $S$, contributing $\zeta d/4$ to
the return mass. With weight $4-d$ it enters an attached region; by
Lemma~\ref{lem:region} its $\zeta$-weighted mass of ever reaching that region's two
terminals (i.e.\ returning to $S$) is at most $h$, contributing at most
$\zeta(4-d)h/4$. Using $d\le2$, $h<1$, the $\zeta$-weighted mass of return on one excursion is at most
\begin{equation*}
  \frac{\zeta}{4}\bigl(d+(4-d)h\bigr)
  \;\le\; \frac{\zeta}{4}(2+2h) = \rho.\qedhere
\end{equation*}
\end{proof}
 
A history with exactly $j$ returns to $S$ consists of $j$ excursions of type (b) followed by one of type (a); the latter is absorbed inside one attached
region of size at most $N/2$, so by the inductive hypothesis it contributes $\zeta$-weighted mass at most $\zeta\,C(N/2)^\kappa$ (the factor $\zeta$ accounts for the edge crossing from $S$ into the region). By
Lemma~\ref{lem:excursion}, each type-(b) excursion contributes a factor at most $\rho$. Summing the geometric series over $j\ge0$,
\[
  \mathbb E_S[\zeta^\tau] \;\le\; \sum_{j\ge0}\rho^j\cdot\zeta\,C(N/2)^\kappa
  \;=\; \frac{\zeta}{1-\rho}\,C(N/2)^\kappa \;=\; \Lambda\,C(N/2)^\kappa. \tag{$\ast$}
\]

Recall $\Lambda:=\zeta/(1-\rho)$ and $\kappa:=\log_2 (1+\Lambda)$. Then
by $(\ast)$, for a walk started on $S$,
\[
  \mathbb E_S[\zeta^\tau] \;\le\; \Lambda\,C(N/2)^\kappa
  = C\cdot\frac{\Lambda}{2^\kappa}\cdot N^\kappa
  \;\le\; C\cdot\frac{1+\Lambda}{2^\kappa}\cdot N^\kappa = CN^\kappa,
\]

If instead the walk starts inside an attached region, Lemma~\ref{lem:region} bounds the $\zeta$-weighted mass of ever reaching $S$ by $h<1$, and conditioning on that event and applying the bound just proved gives
\[
  \mathbb E_v[\zeta^\tau] \;\le\; (1+h\Lambda)\,C(N/2)^\kappa
  \;\le\; (1+\Lambda)\,C(N/2)^\kappa = CN^\kappa.
\]
This holds for every starting vertex and closes the induction on $N$.

\begin{proof}[Proof of~\Cref{thm:threshold-four}]
Make every vertex of
outdegree at most $3$ absorbing.  At every remaining vertex $u$, each
outgoing edge is chosen with probability $1/\outdeg(u)\le1/4$.  The
resulting chain is exactly of the type just analyzed, and its absorption
time is $\tau_4$ (the claim is trivial if the starting vertex is already
low). Markov's inequality
gives the tail bound
\[
  \Pr(\tau_4>m) = \Pr(\zeta^{\tau_4}>\zeta^m) \le Cn^\kappa\zeta^{-m},
\]
and Jensen's inequality gives $\E_s[\tau_4]=O(\log n)$.
\end{proof}

\begin{corollary}[Dynamic outerplanar orientation]
\label{cor:dynamic-four}
Starting from the empty graph, a simple outerplanar graph can be
maintained under insertions and deletions that preserve outerplanarity
so that every outdegree is at most $4$.

For every update, conditioned on the complete state before that update,
the expected update time and expected number of flipped edges are
$O(\log n)$.  Moreover, for every $0<\delta<1$, both are
$O(\log(n/\delta))$ with probability at least $1-\delta$.
Consequently, over every polynomial-length update sequence, all update
times and numbers of flipped edges are $O(\log n)$ simultaneously with
high probability.
\end{corollary}

\begin{proof}
Apply the repair rule of \Cref{sec:randomwalk} with $\Delta=4$.  A
deletion needs no repair.  After an insertion, only the chosen tail $s$
can have outdegree $5$.  If this happens, run the uniform directed walk
from $s$ until it first reaches a vertex of outdegree at most $3$.

By \Cref{thm:threshold-four}, the sampled walk has length $O(\log n)$
in expectation and $O(\log(n/\delta))$ with probability at least
$1-\delta$.  Its retained path is no longer than the sampled walk, and
flipping that path restores maximum outdegree $4$.  By the
implementation described in \Cref{sec:randomwalk}, recording the walk,
extracting its retained path, and flipping it take time linear in the
sampled length.

Since \Cref{thm:threshold-four} holds for every current orientation,
these bounds remain valid after conditioning on the complete preceding
history.  The simultaneous high-probability guarantee over a
polynomial-length update sequence follows from the same conditional
union-bound argument as in the proof of \Cref{cor:dynamic}.
\end{proof}

Thus \Cref{thm:threshold-four} proves the upper bound assertion of
\Cref{thm:intro-tight-threshold}, while
\Cref{cor:dynamic-four} proves \Cref{thm:intro-outerplanar}.

\section{Extension to \texorpdfstring{$K_{2,t}$}{K\_2,t}-minor-free graphs}
\label{sec:k2t-extension}
Recall that outerplanar graphs are $K_{2,3}$-minor-free, but the converse is false because outerplanarity additionally excludes $K_4$ as a minor. Thus, even the case $t=3$ below is already a genuine extension of the graph family.

The extension in this section is based on Layer I. That argument uses outerplanarity only through \Cref{lem:path-count}, which bounds the number of simple paths of a given length between two fixed vertices. We prove a corresponding fixed-endpoint path bound for $K_{2,t}$-minor-free graphs and then apply the same random-walk analysis with this new bound. This yields an $O(\log n)$ hitting-time bound with an outdegree threshold depending only on $t$. The sharper Layer-II analysis uses more detailed two-terminal structure and is not needed for this extension. Throughout this section, $t\ge2$ denotes the excluded-minor parameter; we use $k$ for a path length, $m$ for a random walk length, and $L$ for a closed-walk length bound.

Related global copy-counting results for fixed forests in
$K_{s,t}$-minor-free graphs are due to Huynh and Wood
\cite{HW22}. Here we need an explicit fixed-endpoint bound; the
following direct argument derives one from Menger's theorem and the
Catalan recurrence. Let $C_j:=\frac{1}{j+1}\binom{2j}{j}$ be the $j$th
{\em Catalan number}. We use the recurrence
\begin{equation} \label{eq:catalan}
C_{k-1}=\sum_{i=1}^{k-1}C_{i-1}C_{k-i-1}
\qquad\text{for every }k\ge2.
\end{equation}

\begin{lemma}[Path count for $K_{2,t}$-minor-free graphs]
\label{lem:k2t-path-count}
Let $t\ge2$, and let $G$ be a simple $K_{2,t}$-minor-free graph.
For any two distinct vertices $a,b$ and every integer $k\ge1$, the
number of simple $a$--$b$ paths with exactly $k$ edges is at most $C_{k-1}(t-1)^{k-1}$.
In particular, this number is at most $\bigl(4(t-1)\bigr)^{k-1}$.
The same bounds hold for {\em directed} simple $a$--$b$ paths in every
orientation of $G$.
\end{lemma}

\begin{proof}
For a graph $H$ and vertices $u,v\in V(H)$, let $N_j^H(u,v)$ denote the number of simple $u$--$v$ paths in $H$ with exactly $j$ edges. We prove that $N_k^G(a,b)\le C_{k-1}(t-1)^{k-1}$ by induction on $k$, simultaneously for all simple $K_{2,t}$-minor-free graphs and all pairs of distinct endpoints.

The argument has one idea at its core: since $G$ excludes a $K_{2,t}$ minor, it cannot contain $t$ internally disjoint $a$--$b$ paths, so Menger's theorem supplies a small separator of at most $t-1$ vertices meeting every $a$--$b$ path. Splitting each path at the first separator vertex it reaches turns the count of long paths into a sum of products of counts of shorter paths; and that sum, remarkably, is exactly the Catalan recurrence. We now carry this out in detail.

\smallskip
\noindent\textit{Base case.} For $k=1$, the simplicity of $G$ allows for at most one edge $ab$, hence at most one path, matching $C_0(t-1)^0 = 1$.
 
\smallskip
\noindent\textit{Inductive step.} Let $k \ge 2$, and assume the bound for every smaller positive path length. If $ab \notin E(G)$, set $G' := G$; otherwise let $G'$ be $G$ with the edge $ab$ deleted. A simple $a$--$b$ path with more than one edge cannot use $ab$, so
$N^G_k(a,b) = N^{G'}_k(a,b)$, and $G'$ remains $K_{2,t}$-minor-free.
 
\emph{No $t$ disjoint paths.} We claim $G'$ has no $t$ internally vertex-disjoint $a$--$b$ paths. If it did, say $P_1,\dots,P_t$, then since $a,b$ are non-adjacent in $G'$ each $P_j$ has at least one internal vertex; contracting the internal vertices of each $P_j$ to a single point produces $t$ distinct vertices, each adjacent to both $a$
and $b$, i.e.\ a $K_{2,t}$ minor --- a contradiction.
 
\emph{A small separator.} If $a,b$ lie in different components of $G'$, then $N^{G'}_k(a,b)=0$ and we are done. Otherwise, by Menger's theorem there is a set $X \subseteq V(G')\setminus\{a,b\}$ with $|X|\le t-1$ that meets every $a$--$b$ path in $G'$.
 
\emph{Splitting at the first hit.} Assign each simple $a$--$b$ path $P$ with $k$ edges to the first vertex $x \in X$ it meets, traversing $P$ from $a$ to $b$; say $x$ is reached after $i$ edges, where $1 \le i \le k-1$. Splitting $P$ at $x$ yields an $i$-edge simple $a$--$x$ path together with a $(k-i)$-edge simple $x$--$b$ path, and this pair determines $P$ uniquely. Hence the number of paths assigned to $x$ with split position $i$ is at most $N^{G'}_i(a,x)\, N^{G'}_{k-i}(x,b)$ --- possibly an overcount, since some such pairs may fail to glue into a simple path or may reach an earlier vertex of $X$ first, but an overcount only preserves the direction of the inequality we need. Summing over $x \in X$ and $1 \le i \le k-1$, and applying the inductive hypothesis to each factor:
\begin{align*}
  N_k^{G'}(a,b)
  &\le
  \sum_{x\in X}\sum_{i=1}^{k-1}
    N_i^{G'}(a,x)\,N_{k-i}^{G'}(x,b) 
  ~\le~
  |X|\sum_{i=1}^{k-1}
    C_{i-1}(t-1)^{i-1}
    C_{k-i-1}(t-1)^{k-i-1} \\
  &\le
  (t-1)^{k-1}
  \sum_{i=1}^{k-1}C_{i-1}C_{k-i-1} 
  ~=~
  C_{k-1}(t-1)^{k-1},
\end{align*}
The second inequality uses the induction hypothesis, and the final
equality is the Catalan recurrence (\Cref{eq:catalan}). This closes the induction.

\smallskip
Since $C_j \le 4^j$ for every $j \ge 0$, we obtain $N^G_k(a,b) \le \big(4(t-1)\big)^{k-1}$.
Forgetting edge directions maps every directed simple $a$--$b$ path to a distinct undirected simple $a$--$b$ path, so the same bounds hold for directed simple paths in every orientation of $G$.

\end{proof}

We next isolate the part of the random walk argument that uses the path
count.  

\begin{lemma}[From path counts to random walks]
\label{lem:path-count-transfer}
Let $B\ge1$ be an integer, and let $G$ be an $n$-vertex simple graph such
that, for every two distinct vertices $a,b$ and every $k\ge1$, there are
at most $B^k$ simple $a$--$b$ paths with exactly $k$ edges.  Orient the
edges of $G$, set $\Delta:=12B$, and call a vertex {\em high} if its outdegree is at least $\Delta$.  
Starting at any vertex $s$, run the ordinary directed random walk until it reaches a vertex of outdegree less than $\Delta$, and let $\tau_\Delta$ be the number of edges traversed. 
Then, for every integer $m\ge0$, 
\[
\Prob(\tau_\Delta>m)\le3n\,2^{-m}.
\]
\end{lemma}

\begin{proof}

The proof is essentially a parameterized repetition of the arguments of
\Cref{lem:closed,thm:walk}, with the outerplanar path-count bound $4^k$ (Lemma~\ref{lem:path-count}) replaced throughout by the hypothesized bound $B^k$, and the threshold $48$ replaced by $\Delta = 12B$. We record only where the argument uses these two quantities and refer to the proof of Theorem~\ref{thm:walk} for unchanged steps.
 
Define the edge weight $\mathrm{wt}(u\to v) := 2/d^+(u)$ for edges with a high tail. Every property of these weights used there --- outgoing weights at a high vertex sum to $2$, and each edge of a high walk has weight at most $2/\Delta$.

The one change is in bounding the total weight of high closed walks (Lemma~\ref{lem:closed}): each retained sub-path of length $\ell$ was bounded by $4^\ell$ possible choices in the outerplanar setting; here it is bounded by $B^\ell$ by hypothesis. For a high vertex $x$, let $S_L(x)$ be the total weight of all high closed walks at $x$ having at most $L$ edges. We prove by retracting the same induction on $L$ that $S_L(x)<2$ for every high vertex $x$. 
Thus, after one fixed first edge, all possible remainders have total weight at most
\[
  \sum_{\ell \ge 2} B^\ell \left(\frac{2}{\Delta}\right)^{\!\ell} 2^{\ell+1}
  \;=\; 2\sum_{\ell \ge 2} \left(\frac{4B}{\Delta}\right)^{\!\ell}
  \;=\; 2\sum_{\ell \ge 2} \left(\frac{1}{3}\right)^{\!\ell} \;=\; \frac{1}{3},
\]
using $\Delta = 12B$ so that $4B/\Delta = 1/3$. The rest of the induction is unchanged. This completes the analogue of \Cref{lem:closed}.

Now fix $m\ge0$.  For $k=0$, there is exactly one retained path, namely
the empty path at $s$.  For $k\ge1$, there are at most $nB^k$ retained
paths with $k$ edges: at most $n$ choices for the last vertex and at most $B^k$ paths to each fixed last vertex. The total weight of all high walks of length $m$ starting at $s$ is therefore less than
\[
  \sum_{k=0}^{m} n B^k \left(\frac{2}{\Delta}\right)^{\!k} 2^{k+1}
  \;<\; 2n \sum_{k \ge 0} \left(\frac{4B}{\Delta}\right)^{\!k}
  \;=\; 2n \sum_{k \ge 0} \left(\frac{1}{3}\right)^{\!k} \;=\; 3n.
\]

By the same calculation as in Observation~\ref{ob:basic}(3), this total weight is $2^m\Prob(\tau_\Delta>m)$.  Hence $\Prob(\tau_\Delta>m)\le 3n\,2^{-m}$, and the tail bound follows.

\end{proof}

Set $B_t:=4(t-1)$ and $\Delta_t:=12B_t=48(t-1)$.
By \Cref{lem:k2t-path-count}, between two fixed vertices of a
$K_{2,t}$-minor-free graph there are at most
$B_t^{k-1}\le B_t^k$ 
simple paths with exactly $k$ edges.  Applying
\Cref{lem:path-count-transfer} gives the following theorem.

\begin{theorem}[Random walk in $K_{2,t}$-minor-free graphs]
\label{thm:k2t-walk}
Let $t\ge2$, and let $G$ be an $n$-vertex simple
$K_{2,t}$-minor-free graph whose edges have been oriented.  Starting at
any vertex, the ordinary directed random walk reaches a vertex of
outdegree less than $\Delta_t=48(t-1)$
after $O(\log n)$ steps in expectation.
More precisely, if $\tau_{\Delta_t}$ is
the number of steps, then $\E[\tau_{\Delta_t}]\le\left\lceil\log_2(3n)\right\rceil+2$, and, for every integer $m\ge0$, $\Prob(\tau_{\Delta_t}>m)\le3n\,2^{-m}$, and, for every $0<\delta<1$,
$\Prob(\tau_{\Delta_t}>
    \left\lceil\log_2\frac{3n}{\delta}\right\rceil)
  \le\delta$.
\end{theorem}

\begin{proof}
By \Cref{lem:k2t-path-count}, for every $k\ge1$, the number of
$k$-edge simple paths between two fixed vertices is at most
$B_t^{k-1}\le B_t^k$, where $B_t=4(t-1)$. Applying
\Cref{lem:path-count-transfer} with $B=B_t$ and
$\Delta=\Delta_t=12B_t=48(t-1)$ gives the tail bound. The expectation and
high-probability bounds follow from the same tail calculations as in the
proof of \Cref{thm:walk}.
\end{proof}

\begin{corollary}[Dynamic orientation of $K_{2,t}$-minor-free graphs; Restatement of~\Cref{thm:intro-k2t}]
\label{cor:k2t-dynamic}
For every fixed integer $t\ge2$, starting from the empty graph, a simple
$K_{2,t}$-minor-free graph can be maintained under insertions and
deletions that preserve $K_{2,t}$-minor-freeness so that every outdegree
is at most $48(t-1)$.

For every update, conditioned on the complete state before that update,
the expected update time and expected number of flipped edges are
$O(\log n)$.  Moreover, for every $0<\delta<1$, both are
$O(\log(n/\delta))$ with probability at least $1-\delta$.  Thus,
over every polynomial-length update sequence, all update times and
numbers of flipped edges are $O(\log n)$ simultaneously with high
probability.
\end{corollary}

\begin{proof}
Apply the repair rule of \Cref{sec:randomwalk} with
$\Delta=\Delta_t=48(t-1)$.  A deletion needs no repair.  After an insertion, only the
chosen tail can have outdegree $\Delta_t+1$.  If this happens, run the directed
random walk until it reaches a vertex of outdegree less than $\Delta_t$,
extract the retained path, and flip it.  The first outdegree decreases
to $\Delta_t$, the last increases to at most $\Delta_t$, and every internal
outdegree is unchanged.

By \Cref{thm:k2t-walk}, the sampled walk has length $O(\log n)$ in
expectation and $O(\log(n/\delta))$ with probability at least
$1-\delta$.  The retained path is no longer than the walk.  Since $t$ is
fixed, $\Delta_t$ is a constant, so the representation from
\Cref{sec:randomwalk} supports every walk step and edge flip in
constant time.  This proves the per-update bounds, even after conditioning
on the complete preceding history.
The high-probability guarantee over a polynomial-length
update sequence follows from the same conditional union-bound argument as
in the proof of \Cref{cor:dynamic}.
\end{proof}
This proves \Cref{thm:intro-k2t}.

\noindent\textbf{Remark.}
For $t=3$, the general argument gives the outdegree bound $96$.  For
outerplanar graphs, the sharper path-counting bound in
\Cref{lem:path-count} gives $48$. However, we did not try to optimize the constants, and so the general extension
can replace the specialized outerplanar analysis, if one is willing to sacrifice a small constant factor (namely, 2) in the outdegree bound.

\paragraph{Optimizing constants.}
We can change the parameters to achieve $\Delta=\ceil{7B/2}$ 
with probability $\Prob(\tau_\Delta>m)\le3n\,\left(1+\frac{1}{125}\right)^{-m}$
using the following changes.
First, define $\wt(u\to v):=(1+\frac{1}{125})\cdot\frac{1}{\outdeg(u)}$.
Second, show that the total weight of all high closed walks at $x$ with at most $L$ edges is less than $3/2$.
As before, the total weight after one fixed edge is at most
\begin{align*}
  \sum_{\ell\ge2}
    B^\ell
    \left(\frac{126/125}{\Delta}\right)^\ell
    \left(\frac{3}{2}\right)^{\ell+1}
  &=
  \frac{3}{2}\sum_{\ell\ge2}\left(\frac{54}{125}\right)^\ell < \frac{1}{2\left(1+\frac{1}{125}\right)}~.
\end{align*}
Therefore, the total weight for all closed walks starting at $x$ with at most $L$ edges is less than $1+\frac{1}{2(1+\frac{1}{125})}\cdot \left(1+\frac{1}{125}\right)=\frac{3}{2}$.
Hence, the total weight of all high walks of length $m$ is bounded by
\begin{align*}
  \sum_{k=0}^{m}
    nB^k
    \left(\frac{126/125}{\Delta}\right)^k
    \left(\frac{3}{2}\right)^{k+1}
  &<
  \frac{3}{2}n\sum_{k\ge0}\left(\frac{54}{125}\right)^k
  <
  3n.
\end{align*}
With the appropriate modification to \Cref{ob:basic}(3), the total weight is
$\left(1+\frac{1}{125}\right)^m\Prob(\tau_\Delta>m)$ and the improvement follows.

\section{Dynamic Distributed Networks}
\label{sec:dynamic-distributed}

The random-walk repair rule is not just fast, it is also \EMPH{local}. In this section we show that it naturally transfers to a basic dynamic distributed message-passing model.
 
\paragraph{Model.} We use the standard \EMPH{local-wakeup CONGEST model}; see, e.g., \cite{PPS16,CHK16,KS18,AOSS18,ALS22}.
There is a fixed set $V$ of $n$ processors, and the undirected communication graph starts with no edges. Each update inserts or deletes a single edge, and ---
critically --- only its two endpoints wake up; every other processor stays asleep until a message reaches it. Awake processors communicate in synchronous rounds, and in each round every awake processor may send an $O(\log n)$-bit message to any subset of its current neighbors, with possibly different messages sent to different
neighbors. Each processor stores the current orientations of its incident edges and
its list of outgoing edges.

Updates occur sequentially and are sufficiently spaced that the repair
following one update finishes before the next update occurs.  This makes
worst-case bounds particularly relevant: they give a uniform bound on the
recovery time, and hence on the required spacing between consecutive
updates.  The \emph{round complexity} is the number of rounds until a
valid orientation is restored, and the \emph{message complexity} is the
total number of messages sent during the update procedure.  Since deleting
an edge cannot violate the outdegree bound, deletions require no repair;
thus our algorithm applies to both graceful and abrupt edge deletions.

\bigskip
\noindent A centralized algorithm can employ global data structures, such as a queue, counters, tables indexed by vertex IDs, etc. None of this exists in this model. Thus, a repair has to operate using only messages passed between neighbors.
The following lemma is the key assertion behind
\Cref{cor:intro-distributed}.

\begin{lemma}[Distributed implementation]
\label{lem:distributed-implementation}
Fix an outdegree threshold $\Delta$.  If a random walk repair samples a
walk of length $L$, counting repetitions, then the repair can be
implemented in the local-wakeup CONGEST model using $O(L+1)$ rounds and
$O(L+1)$ messages.
\end{lemma}
    
\begin{proof}
Consider an edge insertion.  The two endpoints orient the new edge
according to any fixed local rule, say from the smaller-ID endpoint to
the larger-ID endpoint.  If its tail does not exceed outdegree $\Delta$,
there is nothing more to do.

\emph{The idea:} When an insertion overloads a vertex
$s$, a single token is spawned at $s$ and walks the graph exactly like the centralized random walk --- sampling a uniformly random outgoing edge at each step.
As the walk proceeds, the token maintains its current loop-erased path
from $s$.  Every processor on this path stores a mark and a predecessor
pointer:
 
\begin{itemize}
  \item Each vertex the token visits marks itself and records which vertex sent the
  token (its predecessor on the path).
  \item If the token arrives at an already-marked vertex, it has just closed a loop. It backtracks along the predecessor pointers to erase that loop --- clearing marks as it goes --- then resumes the walk from there.
  \item Once the token reaches a vertex $z$ with outdegree below the threshold, the
  marked trail is exactly the retained path from $s$ to $z$. The token walks back along it one last time, flipping every edge as it passes. 
  Flipping the path in reverse order gives $z$ one more outgoing edge, $s$ one fewer, and leaves every internal outdegree unchanged.
\end{itemize}

\emph{Cost.} If the sampled walk has length $L$ --- counting every repeated vertex
--- the entire repair costs $O(L)$ rounds and $O(L)$ messages: Every path edge traversed while erasing a loop can be charged to the earlier random walk step that first placed that edge on the current path, and each loop erasure incurs only one additional traversal of its closing edge.  Hence all loop erasures together use
$O(L)$ transmissions.  Finally, the remaining path has length at most
$L$, so the final backtracking also costs $O(L)$ transmissions.  Since
the token moves by at most one edge per round and every message has
$O(\log n)$ bits, the total cost is $O(L+1)$ rounds and $O(L+1)$ messages.
\end{proof}

Combining \Cref{lem:distributed-implementation} with the random walk
bounds of \Cref{thm:threshold-four,thm:k2t-walk} yields
\Cref{cor:intro-distributed}.  In particular, for every update and
conditioned on the complete preceding history, both the expected round
complexity and the expected message complexity are $O(\log n)$.  Since
the corresponding tail bounds also hold after conditioning on the
preceding history, a union bound implies that, throughout every
polynomial-length execution, all updates use $O(\log n)$ rounds and
messages simultaneously w.h.p.  Thus the guarantees hold against an
adaptive adversary.

\paragraph{Relation to previous work.}
A small centralized update time or recourse bound does not by itself
yield an efficient implementation in the local-wakeup model.  Even the
BF algorithm coordinates its reset cascade using a global queue, and
centralized dynamic algorithms often rely on heavier global data structures  or rebuilding procedures. Such \EMPH{global
coordination} can be considerably more expensive in a distributed
network, particularly in terms of the number of messages; this general
difficulty is discussed in~\cite{ALS22}.  Our algorithm avoids it
altogether: there is only one active token, and every message can be
charged to a traversal made by that token.  We therefore restrict the
comparison below to the two previous works that explicitly studied
dynamic distributed edge orientation~\cite{PPS16,KS18}.

Parter, Peleg, and Solomon~\cite{PPS16} work in the LOCAL local-wakeup model and
maintain an $O(\alpha+\log^* n)$-orientation using $O(\log^* n)$
amortized rounds.  Although their main focus is on amortized complexity,
their algorithm also has $O(\log n)$ worst-case round complexity.  They
do not analyze message complexity; moreover, the LOCAL model permits messages
of unbounded size.

Kaplan and Solomon~\cite{KS18} work in the CONGEST local-wakeup model and, for
$\Delta=O(\alpha)$, obtain $O(\log n)$ amortized round and message
complexities.  The remark following their Theorem~2.2 sketches how a
carefully truncated exploration could yield $O(\log n)$ worst-case
round complexity, but the details are omitted.  Moreover, this
truncation does not provide a worst-case message-complexity bound.
In contrast, our random walk implementation gives per-update tail bounds
for both measures: in the graph families considered here, every update
uses $O(\log n)$ rounds and messages simultaneously w.h.p.\ over every
polynomial-length execution.

\section{Some Limitations}
\label{sec:limitations}

For a threshold $\Delta$, call a vertex \emph{low} if its outdegree is less
than $\Delta$.  For a uniform directed random walk $X_0=s,X_1,\ldots$, write
\[
  \tau_\Delta:=\min\{j\ge0:\outdeg(X_j)<\Delta\}
\]
for the first time the walk reaches a low vertex.

\subsection{A Polynomial Lower Bound for
\texorpdfstring{$\Delta=3$}{Delta=3} in Outerplanar Graphs}
\label{sec:outerplanar-three-lower}

We first consider $\Delta=3$.  Even in an outerplanar graph, a walk started at
the unique overloaded vertex may need polynomially many steps to find a
low vertex.
We now show the first lower bound of~\Cref{thm:intro-tight-threshold}.

\begin{theorem}[Outerplanar lower bound for $\Delta=3$]
\label{thm:outerplanar-three-lower}
Let $\gamma:=\log_2(7/6)$.  For every $n$ that is sufficiently large, there is an
orientation $\vec G$ of a connected $n$-vertex outerplanar graph and a
vertex $s$ such that $\outdeg(s)=4$, every other vertex has outdegree $1$
or $3$.
Then,
$\E_s[\tau_3]=\Omega(n^{\gamma})$, and
$\tau_3=\Omega(n^{\gamma/2})$ with probability
$1-O(n^{-\gamma/2})$.
Moreover, deleting one outgoing edge of $s$ leaves a connected outerplanar
graph of maximum outdegree $3$.
\end{theorem}
\noindent {\bf Remark.} There always (with probability 1) \EMPH{exists} in the orientation $\vec G$ an $O(\log n)$-length path from $s$ (as well as from any other vertex) to a low-outdegree vertex; see \Cref{outdeg3offline}. However, \Cref{thm:outerplanar-three-lower} shows that the random walk will not provide such a path  with probability $1-O(n^{-\gamma/2})$.

The construction uses a single gadget.  Let $ab$ be an existing edge,
and distinguish $a$ as the \emph{back terminal} and $b$ as the
\emph{forward terminal}, denoting it by $(a,b)$.
The wedge $W_1(a,b)$ adds a new vertex $x$,
called its \emph{top}, and the directed path $a\to x\to b$.  For
$k\ge2$, the wedge $W_k(a,b)$ adds a fresh top $x$ and the same directed
path, and then attaches fresh copies of $W_{k-1}(x,a)$ and
$W_{k-1}(x,b)$ along the edges $(x,a)$ and $(x,b)$.  The wedge does not add
another copy of its base edge $(a,b)$.
See~\Cref{fig:wedge}.

\begin{figure}[htbp]
\centering
\begin{tikzpicture}[
scale=1.5,
>=Stealth,
font=\small,
vertex/.style={
circle,
draw,
fill=white,
inner sep=2pt
}
]

\begin{scope}[xshift=-3cm]
\node[vertex] (a) at (0,0) {${\;}a{\;}$};
\node[vertex] (b) at (1,0) {${\;}b{\;}$};
\node[vertex] (x) at (1/2, 1) {${\;}x{\;}$};

\draw (a) edge[-] (b);
\draw (a) edge[->] (x) (x) edge[->] (b);

\draw (1/2,-1/4) node[anchor=north] {$W_1(a,b)$};
\end{scope}

\begin{scope}[xshift=0cm]
\node[vertex] (a) at (0,0) {${\;}a{\;}$};
\node[vertex] (b) at (1,0) {${\;}b{\;}$};
\node[vertex] (x) at (1/2, 1) {${\;}x{\;}$};
\node[vertex] (x1a) at (-1+1/2, 1) {${\;\;\;\;}$};
\node[vertex] (x1b) at (+1+1/2, 1) {${\;\;\;\;}$};

\draw (a) edge[-] (b);
\draw (a) edge[->] (x) (x) edge[->] (b);
\draw (x) edge[->] (x1a) (x1a) edge[->] (a);
\draw (x) edge[->] (x1b) (x1b) edge[->] (b);
\draw (1/2,-1/4) node[anchor=north] {$W_2(a,b)$};
\end{scope}

\begin{scope}[xshift=4cm]
\node[vertex] (a) at (0,0) {${\;}a{\;}$};
\node[vertex] (b) at (1,0) {${\;}b{\;}$};
\node[vertex] (x) at (1/2, 1) {${\;}x{\;}$};
\node[vertex] (x1a) at (-1+1/2, 1) {${\;\;\;\;}$};
\node[vertex] (x1b) at (+1+1/2, 1) {${\;\;\;\;}$};
\node[vertex] (x2a) at (0, 2) {${\;\;\;\;}$};
\node[vertex] (x2b) at (1, 2) {${\;\;\;\;}$};
\node[vertex] (x3a) at (-1, 0) {${\;\;\;\;}$};
\node[vertex] (x3b) at (2, 0) {${\;\;\;\;}$};

\draw (a) edge[-] (b);
\draw (a) edge[->] (x) (x) edge[->] (b);
\draw (x) edge[->] (x1a) (x1a) edge[->] (a);
\draw (x) edge[->] (x1b) (x1b) edge[->] (b);
\draw (x1a) edge[->] (x3a) (x3a) edge[->] (a);
\draw (x1a) edge[->] (x2a) (x2a) edge[->] (x);
\draw (x1b) edge[->] (x3b) (x3b) edge[->] (b);
\draw (x1b) edge[->] (x2b) (x2b) edge[->] (x);
\draw (1/2,-1/4) node[anchor=north] {$W_3(a,b)$};
\end{scope}

\end{tikzpicture}
\caption{A depiction of $W_k(a,b)$ for $k\in\{1,2,3\}$.}
\label{fig:wedge}
\end{figure}
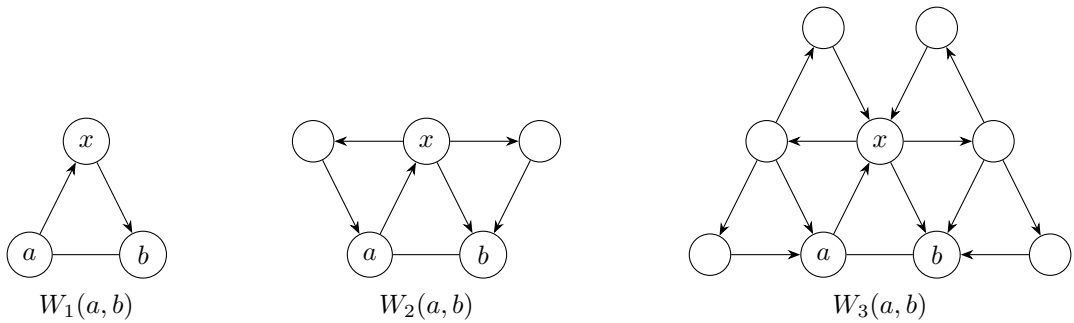

A direct induction shows that $W_k(a,b)$ adds $2^k-1$ vertices, adds one
outgoing edge at $a$ and none at $b$, and gives every new vertex
outdegree $1$ or $3$.  The same induction shows that, whenever $ab$ lies
on the outer face, the wedge can be drawn outside $ab$ while preserving
an outerplanar drawing.

\begin{lemma}[Wedge estimate]
\label{lem:three-wedge}
Start at the top of $W_k(a,b)$ and stop on reaching $a$, $b$, or a low
internal vertex.  Let $u_k$ be the probability of reaching $a$, and let
$r_k$ be the probability of reaching a low internal vertex.  The walk
stops almost surely and, with $\rho:=6/7$,
\[
  u_k\le\frac13,
  \qquad
  r_k\le\rho^{k-1}.
\]
\end{lemma}

\begin{proof}
For $k=1$, the top is low, so $(u_1,r_1)=(0,1)$.  Let $k\ge2$ and write
$u:=u_{k-1}$ and $r:=r_{k-1}$.  By induction, a walk entering either
subwedge eventually reaches one of its terminals or a low internal
vertex.

Consider one visit to the top of $W_k(a,b)$.  The walk returns to the top
with probability $2u/3$, reaches $a$ with probability $(1-u-r)/3$, and
reaches a low vertex with probability $2r/3$; all remaining outcomes
reach $b$.  By induction, $u\le1/3$, so the return probability is at most
$2/9$.  Hence the probability of returning to the top $m$ consecutive
times is at most $(2/9)^m$, and a non-returning visit occurs almost
surely.  Conditioning on the first such visit gives
\[
  u_k=\frac{1-u-r}{3-2u},
  \qquad
  r_k=\frac{2r}{3-2u}.
\]
The first fraction is at most $1/3$, since
$3(1-u-r)\le3-2u$.  Also $3-2u\ge7/3$, and therefore
\[
  r_k\le\frac67r_{k-1}\le\rho^{k-1}.
\]
This completes the induction.
\end{proof}

\begin{proof}[Proof of \Cref{thm:outerplanar-three-lower}]
Fix $h\ge2$.  Build one arm from a fan on $s,v_1,\ldots,v_h$, with spokes
$sv_i$ and rim edges $v_iv_{i+1}$.  Orient every spoke as $v_i\to s$ and
every rim edge as $v_i\to v_{i+1}$.
Attach a wedge $W_h(s,v_1)$, which we call an ``entry wedge'',
and, for each $i<h$, a wedge $W_{h-i}(v_i,v_{i+1})$, which we call a ``forward wedge''.
Take four
copies of the arm and identify their copies of $s$; call the resulting
orientation $\vec G_h$, with underlying graph $G_h$.

\begin{figure}[htbp]
\centering
\begin{tikzpicture}[
scale=0.9,
>=Stealth,
font=\small,
vertex/.style={
circle,
draw,
fill=white,
inner sep=2pt
}
]

\node[vertex] (s) at (0,1) {${\;}s{\;}$};
\node[vertex] (v1) at (2,3) {$v_1$};
\node[vertex] (v2) at (5,4) {$v_2$};
\node[vertex] (v3) at (8,4) {$v_3$};
\node[] (d) at (10,3.5) {$\ddots$};
\node[vertex] (vh) at (12,2) {$v_h$};

\node[] (w0) at (1-0.75,2+0.5) {$W_{h}(s,v_1)$};
\node[] (w1) at (7/2-0.75,7/2+0.5) {$W_{h-1}(v_1,v_2)$};
\node[] (w2) at (13/2+0.25,4+0.75) {$W_{h-2}(v_2,v_3)$};
\node[] (w3) at (9+0.5,3.75+1) {$W_{h-3}(v_3,v_4)$};
\node[] (wh) at (11+1,2.75+0.75) {$W_{1}(v_{h-1},v_h)$};

\draw (v1) edge[->] (s);
\draw (v2) edge[->] (s);
\draw (v3) edge[->] (s);
\draw (d) edge[->] (s);
\draw (vh) edge[->] (s);
\draw (v1) edge[->] (v2);
\draw (v2) edge[->] (v3);
\draw (v3) edge[->] (d);
\draw (d) edge[->] (vh);

\draw (s) edge[->] (w0) (w0) edge[->] (v1);
\draw (v1) edge[->] (w1) (w1) edge[->] (v2);
\draw (v2) edge[->] (w2) (w2) edge[->] (v3);
\draw (v3) edge[->] (w3) (w3) edge[->] (d);
\draw (d) edge[->] (wh) (wh) edge[->] (vh);

\end{tikzpicture}
\caption{Illustration of one arm of the bad instance for the proof of~\Cref{thm:outerplanar-three-lower}.}
\label{fig:fanwedge}
\end{figure}
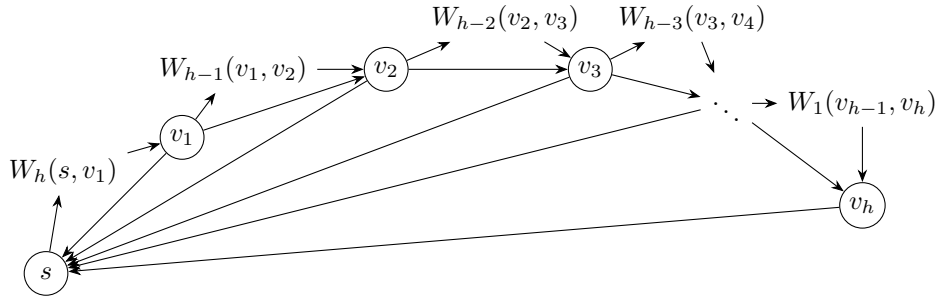

Place $s,v_1,\ldots,v_h$ on the outer face in this order.  Then $sv_1$
and all rim edges are boundary edges.  Draw the wedges outside them and
the four arms in disjoint sectors around $s$.  Thus $G_h$ is outerplanar.
The vertex $s$ has one outgoing edge into each entry wedge.  For $i<h$,
the three outgoing edges of $v_i$ go to $s$, to $v_{i+1}$, and into its
forward wedge, while $v_h$ has only the edge to $s$.  Hence
$\outdeg(s)=4$, and every other outdegree is $1$ or $3$.

Let $x$ be the top of the entry wedge of one arm.  The edge $s\to x$ lies in the
triangle $sxv_1$.  Deleting it leaves $x$ connected to $s$ through $v_1$
and lowers only the outdegree of $s$, from $4$ to $3$.  This is the edge
whose insertion creates the overloaded vertex.

An \emph{excursion} starts when the random walk leaves $s$ and ends when it first
returns to $s$ or reaches a low vertex.  Let $p_h$ be its success
probability.  The four arms are identical, so consider one arm.

For each $i<h$, define $k:=h-i$. During one visit to $v_i$, the walk returns
directly to $s$ with probability $1/3$, returns to $v_i$ through
$W_k(v_i,v_{i+1})$ with probability $u_k/3$, reaches $v_{i+1}$ with
probability $(2-u_k-r_k)/3$, and finds a low vertex in the wedge with
probability $r_k/3$.  Conditioning on the first visit that does not
return to $v_i$, the probabilities of advancing to $v_{i+1}$ and of
finding a low vertex in the wedge are at most
\[
  \frac{2-u_k-r_k}{3-u_k}\le\frac23,
  \qquad
  \frac{r_k}{3-u_k}\le\frac12r_k.
\]
The first bound uses $u_k,r_k\ge0$, while the second uses only
$u_k\le1$.

Thus the probability of reaching $v_i$ from $v_1$ is at most
$(2/3)^{i-1}$.
This is why we constructed the wedge depths to decrease along the arm.  A
successful excursion either finds $v_h$ or finds a low vertex either in the
entry wedge, or in a forward wedge.  By
\Cref{lem:three-wedge}, using $\rho=6/7$,
\begin{align*}
  p_h
  &\le \left(\frac23\right)^{h-1}+r_h
      +\frac12\sum_{i=1}^{h-1}
        \left(\frac23\right)^{i-1}r_{h-i} \\
  &\le 2\rho^{h-1}
      +\frac12\rho^{h-2}\sum_{i=1}^{h-1}\left(\frac79\right)^{i-1} \\
  &\le 2\rho^{h-1}+\frac94\rho^{h-2}
   <5\rho^{h-1}.
\end{align*}
Here $2/3<\rho$ and $(2/3)/\rho=7/9$.

Each arm has $h$ fan vertices and one wedge of every depth $1,\ldots,h$,
so
\[
  n_h:=|V(G_h)|
  =1+4\left(h+\sum_{j=1}^h(2^j-1)\right)
  =2^{h+3}-7.
\]
Since $\rho=2^{-\gamma}$, we have $p_h=O(n_h^{-\gamma})$.

After each failed excursion the walk is back at $s$.  Hence the number
$N$ of excursions up to and including the first successful one is
geometric with parameter $p_h$.  Since $\tau_3\ge N$, for every $M\ge1$,
\[
  \E_s[\tau_3]\ge\frac1{p_h}=\Omega(n_h^\gamma),
  \qquad
  \Prob_s(\tau_3<M)\le\Prob(N<M)
  =\sum_{i=1}^{M-1}(1-p_h)^{i-1}p_h
  \le Mp_h.
\]
For $M=\lfloor n_h^{\gamma/2}\rfloor$, we have
$M=\Theta(n_h^{\gamma/2})$ and $Mp_h=O(n_h^{-\gamma/2})$, giving the
claimed tail bound.  For arbitrary $n$, choose $h$ maximal with
$n_h\le n$.  Then $n<n_{h+1}=2n_h+7$, so $n_h=\Theta(n)$.  Add $n-n_h$
leaves at $s$ and orient their edges toward $s$.  The graph remains
outerplanar, and the walk cannot enter the new leaves, so the same bounds
hold for $n$.
We also note that the path from $s$ to a vertex of outdegree $1$ in $W_h(s, v_1)$ has length $h=O(\log n)$.
\end{proof}

\subsection{A Tree-Based Lower-Bound Construction}
\label{sec:sp-walk-lower}

We now give a simpler tree-based construction.  For every fixed $\Delta\ge3$,
it yields a polynomial lower bound in a series-parallel graph,
that is, graphs that have treewidth at most two.
Importantly, series-parallel graphs are planar and further, are $K_4$-minor-free.
When $\Delta=2$, the trees in our construction degenerate into paths,
and the same excursion calculation yields an exponential lower bound in an outerplanar graph.
We state this consequence separately after the proof.

The construction has one distinguished ``reset'' vertex $s$
and several rooted trees.
As before, an \emph{excursion} starts when the walk leaves
$s$ and ends when it next returns to $s$ or first reaches a low vertex.
After the first move from a root to one of its children, the walk
repeatedly faces the same choice: move one level deeper into the tree or
return directly to $s$.
Thus, reaching a low vertex requires one
successful excursion in which no reset edge is chosen.

\begin{theorem}[Long hitting time in series-parallel graphs;
Restatement of~\Cref{thm:intro-series-parallel}]
\label{thm:sp-walk-lower}
For all integers $\Delta\ge3$ and $n\ge2\Delta^3$, there is a simple connected
$n$-vertex graph $G$ of treewidth at most two, an orientation $\vec{G}$ of
$G$, and a designated vertex $s$, such that every outdegree in $\vec{G}$
is either $\Delta$ or $1$. 
Define $\beta_\Delta:=\log_{\Delta-1}\left({\Delta}/\left({\Delta-1}\right)\right)$.
Then,
$\E_s[\tau_\Delta]=\Omega(n^{\beta_\Delta})$, and
$\tau_\Delta=\Omega(n^{\beta_\Delta/2})$ with probability
$1-O(n^{-\beta_\Delta/2})$.
\end{theorem}

\noindent\textbf{Remark.}
We have
$\beta_\Delta=\frac{\ln(1+1/(\Delta-1))}{\ln(\Delta-1)}
\sim\frac{1}{(\Delta-1)\ln(\Delta-1)}$ as $\Delta\to\infty$.
For every constant $\Delta$, the theorem therefore gives a polynomial lower bound
with inverse-polynomial failure probability.

\begin{figure}[htbp]
\centering
\resizebox{\textwidth}{!}{%
\begin{tikzpicture}[
    >=Stealth,
    node/.style={circle, draw, fill=white, minimum size=4.2mm, inner sep=0pt},
    treeedge/.style={-{Stealth[length=2mm]}, black, thick},
    backedge/.style={-{Stealth[length=2mm]}, red, thick},
    sedge/.style={-{Stealth[length=2mm]}, blue, thick},
]

\newcommand{\drawtree}[4]{

    \node[node] (#2) at (#1, 0) {\Large ${\,}#3{\,}$};

    \node[node] (#2a) at ($(#1,0)+(-1.3,-1.4)$) {};
    \node[node] (#2b) at ($(#1,0)+(0.0,-1.4)$) {};
    \node[node] (#2c) at ($(#1,0)+(1.3,-1.4)$) {};
    \draw[treeedge] (#2) -- (#2a);
    \draw[treeedge] (#2) -- (#2b);
    \draw[treeedge] (#2) -- (#2c);
    \node at ($(#1,0)+(0.75,-1.4)$) {$\cdots$};

    \node[node] (#2a2) at ($(#2a)+(0,-1.2)$) {};
    \node[node] (#2b2) at ($(#2b)+(0,-1.2)$) {};
    \node[node] (#2c2) at ($(#2c)+(0,-1.2)$) {};
    \draw[treeedge] (#2a) -- (#2a2);
    \draw[treeedge] (#2b) -- (#2b2);
    \draw[treeedge] (#2c) -- (#2c2);
    \node at ($(#2b2)+(0.75,0)$) {$\cdots$};

    \node at ($(#2a2)+(0,-0.7)$) {$\vdots$};
    \node at ($(#2b2)+(0,-0.7)$) {$\vdots$};
    \node at ($(#2c2)+(0,-0.7)$) {$\vdots$};

    \node[node] (#2a3) at ($(#2a2)+(0,-2.2)$) {};
    \node[node] (#2b3) at ($(#2b2)+(0,-2.2)$) {};
    \node[node] (#2c3) at ($(#2c2)+(0,-2.2)$) {};
    \node at ($(#2b3)+(0.75,0)$) {$\cdots$};

    \node[node] (#2a4l) at ($(#2a3)+(-0.3,-1.1)$) {};
    \node[node] (#2a4r) at ($(#2a3)+(0.3,-1.1)$) {};
    \node[node] (#2b4l) at ($(#2b3)+(-0.3,-1.1)$) {};
    \node[node] (#2b4r) at ($(#2b3)+(0.3,-1.1)$) {};
    \node[node] (#2c4l) at ($(#2c3)+(-0.3,-1.1)$) {};
    \node[node] (#2c4r) at ($(#2c3)+(0.3,-1.1)$) {};

    \draw[treeedge] (#2a3) -- (#2a4l);
    \draw[treeedge] (#2a3) -- (#2a4r);
    \draw[treeedge] (#2b3) -- (#2b4l);
    \draw[treeedge] (#2b3) -- (#2b4r);
    \draw[treeedge] (#2c3) -- (#2c4l);
    \draw[treeedge] (#2c3) -- (#2c4r);


    \draw[backedge] (#2a4l) to[out=-90,in=#4] (s);
    \draw[backedge] (#2a4r) to[out=-90,in=#4] (s);
    \draw[backedge] (#2b4l) to[out=-90,in=#4] (s);
    \draw[backedge] (#2b4r) to[out=-90,in=#4] (s);
    \draw[backedge] (#2c4l) to[out=-90,in=#4] (s);
    \draw[backedge] (#2c4r) to[out=-90,in=#4] (s);

}

\node[node] (s) at (6.5, -10) {\Large${\;\;}s{\;\;}$};

\drawtree{0}{r1}{r_1}{180}
\drawtree{6.5}{ri}{r_i}{90}
\drawtree{13}{rD}{r_\Delta}{0}

\draw[sedge,<-] (r1) to[out=0,in=170] (s);
\draw[sedge,<-] (ri) to[out=200,in=130] (s);
\draw[sedge,<-] (rD) to[out=180,in=20] (s);

\node at (0, 1.0) {\Large$T_1$};
\node at (6.5, 1.0) {\Large$T_i$};
\node at (13, 1.0) {\Large$T_\Delta$};

\node at (9.75, -3.0) {$\cdots$};

\draw[decorate, decoration={brace, amplitude=6pt}] (-2.0,-5.9) -- (-2.0,0) node[midway, xshift=-14pt] {\Large $h$};

\begin{scope}[shift={(14.7,0)}]
    \draw[dashed] (0,0) -- (1.1,0) node[right, align=left] {Root: $\Delta$ children};
    \draw[dashed] (0,-2.9) -- (1.1,-2.9) node[right, align=left] {Other nonleaf vertices:\\ $\Delta-1$ children};

    \draw[sedge] (0,-6.5) -- (1.1,-6.5) node[right, align=left] {Edge from $s$ to\\ each root};
    \draw[treeedge] (0,-7.5) -- (1.1,-7.5) node[right, align=left] {Tree edges\\ (directed down)};
    \draw[backedge] (0,-8.6) -- (1.1,-8.6) node[right, align=left] {Edge from each\\ nonroot vertex\\ back to $s$};
\end{scope}

\end{tikzpicture}%
}
\caption{The orientation $\vec{G}_h$: $\Delta$ disjoint rooted trees of height $h$, with tree edges directed away from the roots, an edge from $s$ to each root, and an edge from every nonroot tree vertex to $s$. Only representative children and edges to $s$ are shown.}
\label{fig:gadget}
\end{figure}
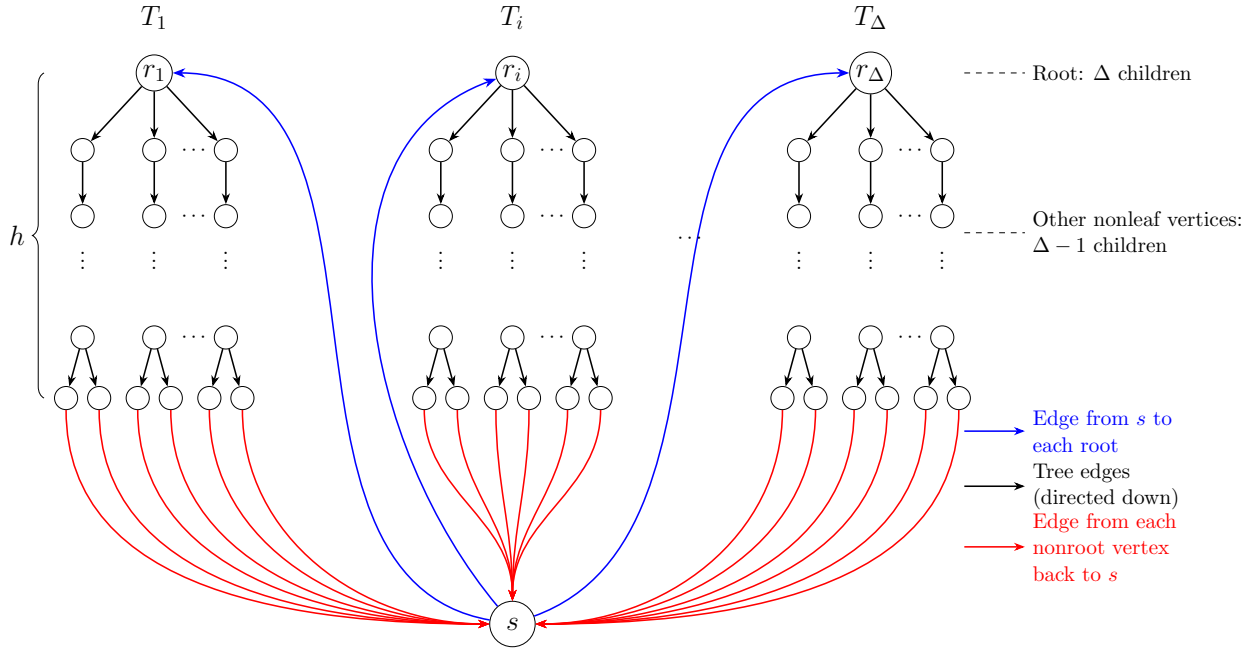

\begin{proof}
For an integer $h\ge2$, let $T_1,\ldots,T_\Delta$ be disjoint rooted trees of height $h$.
In each tree, the root has $\Delta$ children, every other vertex has
$\Delta-1$ children, and all leaves are at distance exactly $h$ from the root.
Let $G_h$ be obtained by adding a new vertex $s$ adjacent to every vertex
of the trees.  Orient (1) the edge from $s$ to the root of each tree,
(2) every tree edge away from the root, and (3) every edge between $s$ and
a nonroot tree vertex $v$ towards $s$.  Denote the resulting orientation
by $\vec{G}_h$.
See~\Cref{fig:gadget}.

Every outdegree is either $\Delta$ or $1$.  Indeed, $s$ has one outgoing edge
to each of the $\Delta$ roots; each root has $\Delta$ outgoing edges to its
children; and every other nonleaf tree vertex has $\Delta-1$ outgoing edges to
its children and one outgoing edge to $s$.  All these vertices have
outdegree $\Delta$, while every leaf has only the outgoing edge to $s$ and
therefore has outdegree $1$.

We next verify that $G_h$ has treewidth at most two.  Let
$B_s:=\{s\}$.  For the root $r_i$ of each tree $T_i$, let
$B_{r_i}:=\{s,r_i\}$ and connect $B_{r_i}$ to $B_s$.  For every nonroot
tree vertex $v$, let $p(v)$ be its parent, set
$B_v:=\{s,p(v),v\}$, and connect $B_v$ to $B_{p(v)}$.
The graph that is constructed from the bags is a tree.
Every edge of $G_h$ is contained in a bag: $(s,r_i)$ is
contained in $B_{r_i}$, while $(p(v),v)$ and $(s,v)$ are contained in
$B_v$.  The set of bags containing $s$ form the entire decomposition tree, and
for every tree vertex $v$, the set of bags containing $v$ are $B_v$ and the set of
bags of its children, which form a connected star.
Thus, this is a tree decomposition of width at most two.

Consider an excursion of the random walk from $s$.  Its first step enters
a root, and from the root it must move to a child.  At every subsequent
nonleaf vertex, one outgoing edge returns to $s$, while the other $\Delta-1$
outgoing edges lead to children.  Hence an excursion reaches a leaf
precisely when the walk chooses a child at each of the $h-1$ subsequent
nonleaf vertices.  Its success probability is
\[
  p_h=\left(\frac{\Delta-1}{\Delta}\right)^{h-1}.
\]

Let $N$ be the number of excursions up to and including the first
successful one.  After every failed excursion the walk returns to $s$,
where the same experiment starts again with fresh random choices.  Thus
$N$ is geometric with success probability $p_h$, and $\E[N]=1/p_h$.
Since every excursion traverses at least one edge, $\tau_\Delta\ge N$, and
therefore $\E[\tau_\Delta]\ge1/p_h$.
For every positive integer $M$, the event $N<M$ means that the first
successful excursion is one of the first $M-1$ excursions.  Hence
$\Prob(\tau_\Delta<M) \le\Prob(N<M) \le Mp_h$.
In particular, for $M_h:=\lfloor p_h^{-1/2}\rfloor$, we get
$\Prob(\tau_\Delta<M_h)\le\sqrt{p_h}$.

It remains to relate $h$ to the number of vertices.  Each tree contains
$1+\Delta\sum_{j=0}^{h-1}(\Delta-1)^j$ vertices, and hence
\begin{equation} \label{eq:nh}
  n_h:=|V(G_h)|
  =
  1+\Delta+\frac{\Delta^2((\Delta-1)^h-1)}{\Delta-2}.
\end{equation}
Since $n_2=\Delta^3+\Delta+1\le2\Delta^3\le n$, choose $h\ge2$ maximal such that
$n_h\le n$.  Increasing the height from $h$ to $h+1$ adds
$\Delta(\Delta-1)^h$ vertices to each of the $\Delta$ trees, and hence
\[
  n_{h+1}-n_h=\Delta^2(\Delta-1)^h.
\]
By maximality of $h$ and the bound
$n_h\ge \Delta^2(\Delta-1)^{h-1}$,
$$
  n
  <
  n_{h+1}
  =
  n_h+\Delta^2(\Delta-1)^h
  \le
  \Delta n_h.
$$
Thus $n_h>n/\Delta$.
Since $\Delta\ge3$, we have $1/(\Delta-2)\le1$, and hence
\Cref{eq:nh} gives
$$
  n_h
  \le
  1+\Delta+\Delta^2((\Delta-1)^h-1)
  \le
  \Delta^2(\Delta-1)^h.
$$
Consequently, since $\Delta^3\le(\Delta-1)^5$ for every $\Delta\ge3$,
\[
  (\Delta-1)^h
  \ge
  \frac{n_h}{\Delta^2}
  >
  \frac{n}{\Delta^3}
  \ge
  \frac{n}{(\Delta-1)^5}.
\]
Recalling the definition of $\beta_\Delta = \log_{\Delta-1}(\Delta/(\Delta-1))$
gives
$$
  \frac1{p_h}
  ~=~
  \left(\frac{\Delta}{\Delta-1}\right)^{h-1} 
  ~=~
  \left(\frac{\Delta-1}{\Delta}\right)^6\bigl((\Delta-1)^{h+5}\bigr)^{\beta_\Delta} 
  ~\ge~
  \left(\frac{\Delta-1}{\Delta}\right)^6n^{\beta_\Delta}
  ~\ge~
  \left(\frac23\right)^6n^{\beta_\Delta}.
$$
Thus $\E[\tau_\Delta]\ge1/p_h=\Omega(n^{\beta_\Delta})$.  As
$M_h=\lfloor p_h^{-1/2}\rfloor\ge p_h^{-1/2}/2$, the same estimate gives
$M_h=\Omega(n^{\beta_\Delta/2})$
and $\sqrt{p_h}=O(n^{-\beta_\Delta/2})$.  Together with
$\Prob(\tau_\Delta<M_h)\le\sqrt{p_h}$, this proves the required bounds.

Finally, if $n_h<n$, add $n-n_h$ new vertices adjacent to $s$,
orienting every new edges towards $s$.  Each new vertex has outdegree $1$,
while all original outdegrees remain unchanged.  Clearly, adding these
vertices does not increase the treewidth beyond two.  Moreover, the walk from $s$ cannot enter a new vertex, so its hitting time is unchanged.
Thus the resulting graph has exactly $n$ vertices and satisfies the same
bounds.
\end{proof}

Based on this construction,
we show the second lower bound of~\Cref{thm:intro-tight-threshold}.

\begin{theorem}[Exponential outerplanar lower bound for $\Delta=2$]
\label{thm:outerplanar-two-walk-lower}
For every sufficiently large $n$, there is a simple connected
$n$-vertex outerplanar graph $G$, an orientation $\vec G$ of $G$, and a
designated vertex $s$ such that $\outdeg(s)=2$ and every outdegree in
$\vec G$ is either $1$ or $2$. 
Then, $\E_s[\tau_2]=2^{\Omega(n)}$, and $\tau_2=2^{\Omega(n)}$ with probability $1-2^{-\Omega(n)}$.
\end{theorem}
\begin{proof}
Use the construction from the proof of
\Cref{thm:sp-walk-lower}, now with $\Delta=2$.  Each $T_i$, viewed as an
undirected graph, is a path with $2h$ edges and with its root in the
middle.  Since $\Delta=2$, there are two such trees: $T_1$ and $T_2$.
Together with $s$, each
forms an outerplanar fan, and the two fans can be drawn in disjoint
regions meeting only at $s$.  Hence $G_h$ is outerplanar.  The vertex $s$
and every nonleaf tree vertex have outdegree $2$, while every leaf has
outdegree $1$.

After entering a root, the walk chooses one of the two sides of its path.
At each of the next $h-1$ nonleaf vertices, it continues toward a leaf
with probability $1/2$ and returns to $s$ with probability $1/2$.
Therefore an excursion succeeds with probability
$p_h=2^{-(h-1)}$.
Applying the geometric excursion calculation from the preceding proof,
for every $M\ge1$,
$$\E_s[\tau_2]\ge 1/{p_h}=2^{h-1}\quad\text{and}\quad
\Prob_s(\tau_2<M)\le Mp_h~.$$
With $M_h:=\lfloor p_h^{-1/2}\rfloor$, this gives
$M_h=2^{\Omega(h)}$,
thus,
$\Prob_s(\tau_2<M_h)\le2^{-(h-1)/2}$.

The graph has $n_h=1+2(2h+1)=4h+3$ vertices.
Since $h=(n_h-3)/4$, these are the claimed exponential bounds
in $n_h$.  For an arbitrary sufficiently large $n$, choose $h$ maximal
with $n_h\le n$ and add $n-n_h<4$ leaves adjacent to $s$, orienting their edges toward $s$.
This preserves outerplanarity and all original outdegrees.  The walk
cannot enter the new leaves, so its hitting time is unchanged.
\end{proof}

\paragraph{Why the construction is not outerplanar for $\Delta\ge3$.}
For every $\Delta\ge3$, the graph $G_h$ from
\Cref{thm:sp-walk-lower} contains a $K_{2,3}$ subgraph.  Indeed, take
any root $a_i$ and any three of its children.  Both $s$ and $a_i$ are
adjacent to all three children.  Since $K_{2,3}$ is not outerplanar,
neither is $G_h$.  This is why the outerplanar lower bound for $\Delta=3$ in
\Cref{thm:outerplanar-three-lower} requires a different construction.

\paragraph{Connection to the dynamic setting.}
The constructions in
\Cref{thm:sp-walk-lower,thm:outerplanar-two-walk-lower} can both be
realized immediately after a single insertion.
Recall this is also the case for~\Cref{thm:outerplanar-three-lower}.
For $\Delta\ge2$, use
$\Delta+1$ trees and initially omit the edge from $s$ to the root of one
tree.  If $v$ is any child of the omitted root, then the two-edge path
from the root through $v$ to $s$ remains, so the graph is still
connected.  Before the insertion, every outdegree is at most $\Delta$.
Insert the missing edge and orient it from $s$ to that root.  Then only
$s$ has outdegree $\Delta+1$, while every other vertex has outdegree $\Delta$ or
$1$.
When $\Delta=2$, the three resulting fans can be drawn in disjoint regions
meeting only at $s$, so the graph remains outerplanar.  When $\Delta\ge3$,
the same tree decomposition of width two applies before and after the
insertion.  Since all the trees are identical, each excursion still
succeeds with probability $p_h$.  Finally, using $\Delta+1$ rather than $\Delta$
trees increases the number of vertices by a factor of at most $\frac{\Delta+1}{\Delta}\le\frac32$.
Hence the same asymptotic lower bounds hold immediately after a single
insertion.  Together with the final assertion of
\Cref{thm:outerplanar-three-lower}, this completes the proof of the
lower bound assertions in \Cref{thm:intro-tight-threshold}.  Moreover,
\Cref{thm:sp-walk-lower} and the post-insertion construction prove
\Cref{thm:intro-series-parallel}.

\subsection{Outdegree 2 Requires \texorpdfstring{$\Omega(n)$}{Ω(n)} Amortized Edge Flips}
\label{sec:outdegree-two-lower-bound}

\Cref{thm:outerplanar-two-walk-lower} shows that, for $\Delta=2$, the
uniform directed random walk can have exponential hitting time even in
an outerplanar graph.  This rules out only that particular repair rule.
We now prove a stronger, algorithm-independent limitation: there is a
fixed update sequence on which every sequence of $2$-orientations, even
an optimal offline sequence chosen with full knowledge of all future
updates, incurs $\Omega(n)$ amortized edge flips.

\ThmRecourse*

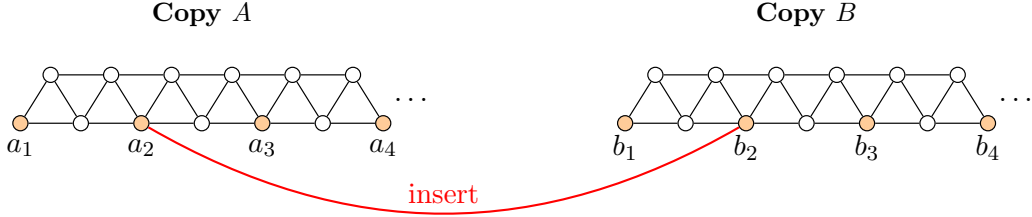
\begin{figure}
    \centering

\begin{tikzpicture}[
    x=8mm,
    y=8mm,
    line width=0.45pt,
    >=Stealth,
    vertex/.style={
        circle,
        draw,
        fill=white,
        minimum size=2mm,
        inner sep=0pt
    },
    port/.style={
        vertex,
        fill=orange!40
    },
    slack/.style={
        vertex,
        fill=blue!25
    },
    bridge/.style={
        thick,
        red
    },
    flip/.style={
        very thick,
        blue
    }
]


\node at (3,1.8) {\small\bfseries Copy $A$};

\foreach \i in {2,4,6}{
    \pgfmathsetmacro\x{\i-1}
    \node[vertex] (A\i) at (\x,0) {};
}

\foreach \i in {1,3,5,7}{
    \pgfmathsetmacro\x{\i-1}
    \pgfmathtruncatemacro{\j}{(\i+1)/2}
    \node[port] (A\i) at (\x,0) {};
    \node[] (N\i) at (\x,-0.4) {$a_{\j}$};
}

\foreach \i in {1,...,6}{
    \pgfmathsetmacro\x{\i-0.5}
    \node[vertex] (U\i) at (\x,0.8) {};
}

\node[] (N) at (6.5,0.4) {$\dots$};

\foreach \i/\j in {1/2,2/3,3/4,4/5,5/6,6/7}
    \draw (A\i)--(A\j);

\foreach \i/\j in {1/2,2/3,3/4,4/5,5/6}
    \draw (U\i)--(U\j);

\foreach \i in {1,...,6}{
    \pgfmathtruncatemacro\next{\i+1}
    \draw (A\i)--(U\i);
    \draw (U\i)--(A\next);
}



\begin{scope}[xshift=8cm]

\node at (3,1.8) {\small\bfseries Copy $B$};

\foreach \i in {2,4,6}{
    \node[vertex] (B\i) at (\i-1,0) {};
}

\foreach \i in {1,3,5,7}{
    \pgfmathtruncatemacro{\j}{(\i+1)/2}
    \node[port] (B\i) at (\i-1,0) {};
    \node[] (M\i) at (\i-1,-0.4) {$b_{\j}$};
}

\foreach \i in {1,...,6}{
    \node[vertex] (V\i) at (\i-0.5,0.8) {};
}

\node[] (M) at (6.5,0.4) {$\dots$};

\foreach \i/\j in {1/2,2/3,3/4,4/5,5/6,6/7} {
    \draw (B\i)--(B\j);
    \draw (B\i)--(V\i);
    \draw (V\i)--(B\j);
}

\foreach \i/\j in {1/2,2/3,3/4,4/5,5/6}
    \draw (V\i)--(V\j);


\end{scope}


\draw[bridge] (A3) to[out=-30,in=-150] node[midway,above]{insert} (B3);

\end{tikzpicture}

\caption{Illustration of the construction used in the proof of Theorem~\ref{thm:intro-outdegree-two}. Each copy is a graph $H_m$, with omitted vertices indicated by dots. The highlighted vertices are the distinguished ports, and the red edge represents a temporary bridge between corresponding
ports. The drawing is schematic: the long portions separating consecutive ports are suppressed.}
\label{fig:outdegree-two-lower-bound}
\end{figure}

\begin{proof}
Let $m:=\lfloor n/2\rfloor$, and let $H_m$ be the graph on
$v_1,\ldots,v_m$ in which $v_i$ and $v_j$ are adjacent whenever
$1\le |i-j|\le 2$.  This is an outerplanar chain of triangles with
$2m-3$ edges.

Insert two disjoint copies $A$ and $B$ of $H_m$, leaving one vertex
isolated if $n$ is odd.  Set $r:=\lfloor(m-1)/24\rfloor$.  In each
copy, choose seven \emph{ports} with indices
$1,1+4r,\ldots,1+24r$, and denote them by
$a_1,\ldots,a_7$ and $b_1,\ldots,b_7$, respectively.  Since
$\operatorname{dist}_{H_m}(v_i,v_j)=\lceil |i-j|/2\rceil$, distinct
ports in the same copy are at distance at least $2r$.
Repeat the following fixed round $m$ times: for $i=1,\ldots,7$, insert
the bridge $(a_i,b_i)$ and then delete it.  Every intermediate graph is
simple and outerplanar. See Figure~\ref{fig:outdegree-two-lower-bound} for an illustration.

Fix a round and consider the orientations at its beginning.
Call a vertex \emph{slack} if its outdegree is at most one.  Every
orientation of $H_m$ with maximum outdegree at most 2 has at most
three slack vertices, since
$\sum_v(2-d^+(v))=2m-(2m-3)=3$.
Thus, the two copies have at most six slack vertices altogether.
Since the ports in
each copy are pairwise at distance at least $2r$, each initial slack
vertex is at distance less than $r$ from at most one port in its copy.

For each $i$, let $z_i\in\{a_i,b_i\}$ be the tail of the bridge after
its insertion has been processed.
That is, the edge between $a_i$ and $b_i$ is oriented away from $z_i$.
The seven vertices $z_1,\ldots,z_7$ are distinct.
Since at most six ports are within
distance less than $r$ of an initial slack vertex, some $z_i$ is at
distance at least $r$ from the initial slack set of its copy.

Fix such a vertex $z:=z_i$.
By renaming, assume $z\in A$.
Compare the internal orientation of $A$ at the beginning of the round with the orientation immediately
after the insertion of $(a_i,b_i)$.
Initially, $z$ has internal
outdegree 2. 
After the insertion, its internal outdegree is at most
one, because the bridge is outgoing from $z$ and its total outdegree
is at most 2.

Let $F$ be the internal edges whose directions differ between these two
orientations.  In every component of $F$, the sum of the outdegree
changes is zero, since each edge of $F$ contributes $-1$ at one
endpoint and $+1$ at the other.  Therefore, because the outdegree of
$z$ decreases, some vertex $w$ in the same component gains outdegree.
Its final internal outdegree is at most 2, so its initial internal
outdegree was at most one.  Thus $w$ was slack at the beginning of the
round and $w\in A$.
The edges of $F$ contain a path from $z$ to $w$.  This path has at
least $r$ edges, since $z$ is at distance at least $r$ from every
initial slack vertex of $A$.
Hence at least $r$ internal edges have changed
direction, and each was flipped at least once during the current
round.  Every round therefore requires at least
$r=\Omega(m)$ edge flips.

Over the $m$ rounds, the number of flips is
$mr=\Omega(m^2)=\Omega(n^2)$.  The setup uses $4m-6$ updates and the
rounds use $14m$ updates, so the entire fixed sequence has
$18m-6=\Theta(n)$ updates.
\end{proof}

\subsubsection{Comparison with Outdegree \texorpdfstring{$3$}{3}} \label{outdeg3offline}

\paragraph{An $O(\log n)$ recourse bound.}
The lower bound above is specific to maximum outdegree $2$.  With
maximum outdegree $3$, the graph $H_m$ has
$3m-(2m-3)=m+3$ units of unused outgoing capacity, rather than $3$.
In fact, the update sequence above requires no edge flips.  In each
copy of $H_m$, orient every edge $(v_i,v_j)$ with $i<j$ as $v_j\to v_i$.
Every vertex then has internal outdegree at most $2$, so every
temporary bridge may be oriented from $A$ to $B$.  All outdegrees
remain at most $3$, and deleting the bridge restores the previous
orientation.

In general, the augmenting path argument of Brodal and Fagerberg
\cite[Lemma~2 and proof of Lemma~3]{BF99} gives
$O(\log n)$ worst-case recourse for every insertion preserving
outerplanarity.  Suppose $(u,v)$ is inserted.  If one endpoint has
outdegree at most $2$, orient the new edge out of that endpoint.
Otherwise, both endpoints have outdegree $3$; choose one of them, say
$u$.  Since outerplanar graphs have arboricity at most $2$, there is a
directed path
$P:u\longrightarrow\cdots\longrightarrow w$
of length at most $\lceil\log_{3/2}n\rceil$, where $w$ has outdegree at
most $2$.  Flip the edges of $P$ and orient the new edge as $u\to v$.
Flipping $P$ decreases the outdegree of $u$ by $1$, increases that of
$w$ by $1$, and leaves all other outdegrees unchanged; the new edge
then restores the outdegree of $u$ to $3$.  Deletions require no flips.

Thus maximum outdegree $3$ can be maintained with $O(\log n)$
worst-case recourse.  This is only a recourse bound: the augmenting
path is short, but finding it may take linear time.

\paragraph{Tightness for an adversarially chosen orientation.}
The $O(\log n)$ upper bound on the worst-case recourse of a single
update is asymptotically tight if the current orientation is chosen
adversarially, even on a forest.  Take two complete rooted ternary
trees of height $h$, orient every edge away from its root, and insert an
edge between the two roots.  Let $x$ be the root that is the tail of
the new edge.  Its outdegree among the old tree edges must decrease
from $3$ to at most $2$.  By conservation of outdegree changes, the
unit lost at $x$ must be transferred along a path of flipped edges to a
vertex whose tree outdegree increases.  Only a leaf can increase its
outdegree, since every nonleaf initially has outdegree $3$.  Hence the
flipped edges contain an $x$-to-leaf path, and at least $h$ old edges
must be flipped.
Conversely, flipping one root-to-leaf path in the tree containing $x$
suffices.  The two trees have $n=3^{h+1}-1$ vertices altogether, and
hence $h=\Theta(\log n)$.

This does not, however, imply an $\Omega(\log n)$ worst-case recourse
lower bound for a dynamic algorithm starting from the empty graph.
Such an algorithm could orient every tree edge toward its root, in
which case inserting the edge between the roots requires no flips.
Thus the example establishes tightness only when the current
orientation itself may be chosen adversarially.

\vspace{-6pt}
\section{Concluding Remarks and Open Questions}
\label{sec:conclusion}
\vspace{-6pt}

We showed that for outerplanar graphs, the thresholds for recourse and for the
uniform-walk paradigm differ by one. Specifically, maintaining outdegree $2$
may require $\Omega(n)$ amortized recourse, even offline,
whereas every insertion preserving outerplanarity admits a repair to outdegree $3$ using $O(\log n)$ recourse, even in the worst-case.  The latter controls the
number of flipped edges, but not the time needed to find them.  Our
random walk algorithm obtains logarithmic update time, against an
adaptive adversary, by allowing outdegree $4$.

\EMPH{The value $4$ is tight} for the uniform directed walk on outerplanar
graphs: \Cref{thm:threshold-four} proves that threshold $4$ has
$O(\log n)$ hitting time in expectation and w.h.p.
Conversely, \Cref{thm:outerplanar-three-lower} gives outerplanar
orientations in which threshold $3$ has $n^{\Omega(1)}$ hitting time in
expectation and w.h.p.; this obstruction can arise
immediately after one insertion into an orientation of
outdegree $3$.  Thus $4$ is the smallest integer $\Delta$ such that, in
every orientation of every outerplanar graph and from every starting
vertex, the uniform directed random walk reaches a vertex of outdegree
less than $\Delta$ within $O(\log n)$ steps, both in expectation and w.h.p.  At threshold $2$, the hitting time can be
exponential in $n$.
(The constant $48$ in the elementary analysis is chosen only for simplicity of presentation.)

This leaves a sharp algorithmic gap.  Can some other update algorithm 
maintain outdegree $3$ with the same $O(\log n)$ per-update
guarantees, in expectation and w.h.p.~against an
adaptive adversary, as our outdegree-$4$ algorithm?
The $O(\log n)$-flip repair shows that there is no recourse barrier, whereas
the threshold-$3$ lower bound shows that the uniform random walk cannot achieve
this.
Even for outdegree $4$, it remains open to obtain a deterministic $O(\log n)$ worst-case update time.

Our positive result extends to $K_{2,t}$-minor-free graphs for every
fixed $t$, but the uniform random walk paradigm does not extend to all
series-parallel graphs.  Indeed, \Cref{thm:sp-walk-lower} shows that,
for every fixed $\Delta\ge3$, the walk can require polynomially many steps
even on graphs of treewidth $2$.  Can this obstruction be overcome by
a more structure-aware \EMPH{nonuniform} walk, using probabilities
based on efficiently maintainable structural information?  Can such a
walk reach a low-outdegree vertex in $O(\log n)$ steps and thereby yield
a constant outdegree dynamic orientation with the same update time
guarantees for series-parallel graphs?  Can the same be achieved for
bounded-treewidth graphs, planar graphs, or, ultimately, every proper
minor-closed graph family?

\section*{AI Disclosure}
The authors used ChatGPT and Claude to assist with drafting, revising, and polishing parts of the manuscript, including portions of the introduction and the presentation of several proofs. ChatGPT was also used in the development of some of the proofs, most notably the proof of \Cref{thm:threshold-four}, as described below.

The authors had previously obtained simpler proofs for larger constant outdegree thresholds. The proof for threshold 4, which requires a substantially different structural argument, was developed through an iterative interaction between the authors and ChatGPT. This interaction played a crucial role in discovering and refining the final argument. The resulting proof of \Cref{thm:threshold-four} was then checked in detail, revised, and fully formalized by the authors.

Every model-generated contribution was reviewed and edited by the authors. All definitions, theorem statements, algorithms, proofs, and final text were approved by the authors, who take full responsibility for the correctness, originality, and integrity of the paper.


\newcommand{\etalchar}[1]{$^{#1}$}

\end{document}